\documentclass[12pt,reqno]{amsart}

\usepackage[margin=1in]{geometry}

\usepackage[group-separator={,}]{siunitx}
\usepackage{svg}
\usepackage{caption}
\usepackage{subcaption}
\usepackage{float}
\usepackage{amsmath,amsthm,amscd,amsfonts,amssymb}
\usepackage{xcolor}
\usepackage{eepic}
\usepackage{enumitem}
\usepackage{hyperref}
\usepackage{cleveref}
\usepackage{mathrsfs}  
\usepackage{tikz}
\usepackage{bm,blkarray}
\usepackage[PostScript=dvips,balance,midshaft,nohug]{diagrams}
\newarrow {Line} {-}{-}{-}{-}{-}
\newarrow{Into}{C}{-}{-}{-}{>}

\DeclareFontFamily{U}{rcjhbltx}{}
\DeclareFontShape{U}{rcjhbltx}{m}{n}{<->rcjhbltx}{}
\DeclareSymbolFont{hebrewletters}{U}{rcjhbltx}{m}{n}
\DeclareMathSymbol{\shin}{\mathord}{hebrewletters}{152}
\DeclareMathOperator{\bad}{\mathrm{bad}}

\swapnumbers  

\newcommand{\QQ}{\mathbb Q}

\newtheorem{thm}[subsection]{Theorem}
\newtheorem{cor}[subsection]{Corollary}
\newtheorem{lem}[subsection]{Lemma}
\newtheorem{prop}[subsection]{Proposition}
\theoremstyle{definition}
\newtheorem{defn}[subsection]{Definition}

\newtheorem{conjecture}[subsection]{Conjecture}
\newtheorem{rem}[subsection]{Remark}
\newtheorem{example}[subsection]{Example}

\numberwithin{equation}{section}

\makeatletter
\renewcommand\section{\@startsection {section}{1}{\z@}%
	{-3.25ex \@plus -1ex \@minus -.2ex}%
	{1.8ex \@plus.2ex}%
	{\centering\normalfont\normalsize\bfseries}}%
\renewcommand\subsection{\@startsection {subsection}{1}{\z@}%
	{-2.0ex \@plus -1ex \@minus -.2ex}%
	{-1.8ex \@plus.2ex}%
	{\normalfont\normalsize\bfseries}}%
\makeatother
\newcommand{\quash}[1]{}

\usepackage[obeyFinal,textsize=footnotesize]{todonotes}

\newcommand{\R}{\mathbb R}
\newcommand{\RR}{\mathbb R}
\newcommand{\CC}{\mathbb C}

\newcommand{\rank}{\mathop{\rm rank}}

\newcommand{\ZZ}{\mathbb Z}
\newcommand{\TT}{\mathbb T}
\newcommand{\cM}{{\mathcal M}}
\newcommand{\cP}{{\mathcal P}}
\newcommand{\tcM}{{\mathcal{M}}}
\newcommand{\cS}{{\mathcal S}}
\newcommand{\cA}{{\mathcal A}}
\newcommand{\cH}{{\mathcal H}}
\newcommand{\tcA}{{\tilde{\mathcal A}}}
\newcommand{\Qin}{Q_{\mathrm{in}}}
\newcommand{\Qout}{Q_{\mathrm{out}}}
\newcommand{\Sig}{{\mathscr S}}  
\newcommand{\nodcnt}{\phi}  
\newcommand{\ta}{{\widetilde\alpha}}

\newcommand{\diag}{\mathop {\rm diag}}
\DeclareMathOperator{\spec}{spec}
\DeclareMathOperator{\Ran}{Ran}

\DeclareMathOperator{\ind}{ind}
\DeclareMathOperator{\Hess}{Hess}
\DeclareMathOperator{\codim}{codim}

\newcommand{\f}{\mathrm{\scriptscriptstyle{free}}}

\usepackage[utf8]{inputenc}

\allowdisplaybreaks

\begin{document}

\title{Smooth critical points of eigenvalues on the torus of magnetic perturbations of graphs}
\author{Lior Alon}
\address{Massachusetts Institute of Technology, Cambridge, MA, USA}

\author{Gregory Berkolaiko}
\address{Texas A\&M University, College
  Station, TX 77843-3368, USA}

\author{Mark Goresky}
\address{Institute for Advanced Study, Princeton, NJ, USA}

\begin{abstract}
Motivated by the nodal distribution universality conjecture for discrete operators on graphs and by the spectral analysis of their maximal abelian covers, we consider a family of Hermitian matrices $h_{\alpha}$ obtained by varying the complex phases of individual matrix elements.  This family is
parametrized by a $\beta$-dimensional torus, where $\beta$ is the first Betti number of the underlying graph.  The eigenvalues of each matrix are ordered, enabling us to treat the $k$-th eigenvalue $\lambda_k$ as a function on the torus.  We classify the smooth critical points of $\lambda_k$, describe their structure and Morse index in terms of the support and nodal count, that is, the number of sign changes 
between adjacent vertices of the corresponding eigenvector.  In general, the families under consideration exhibit critical submanifolds rather than isolated critical points.  These critical manifolds appear frequently and cannot be removed through perturbations.  We provide an algorithmic way of determining all critical submanifolds by investigating finitely many eigenvalue problems: the $2^\beta$ real symmetric matrices $h_\alpha$ in the family under consideration as well as their principal minors.
\end{abstract}	

\maketitle

\section{Introduction}

\subsection{Graph setting and the magnetic torus}

Throughout this paper  $G = (V=[n],E)$ is a simple connected graph on $n$ ordered vertices and first Betti number $\beta = |E|-n+1$.  An $n \times n$ matrix $h$ is {\em supported on $G$} (resp.~{\em strictly supported on $G$}) if, 
 for any $r \ne s$, $h_{rs} \ne 0 \implies (rs) \in E$ (resp.~$h_{rs} \ne 0 \iff (rs) \in E$). 
 Associated to a graph $G$ there are three vector spaces of matrices:
 \begin{itemize}
 \item $\cH (G)$ is the set of \emph{complex hermitian} matrices supported on $G$,
 \item $\cS(G)\subset \cH (G)$ is the set of \emph{real symmetric} matrices supported on $G$,
 \item $\cA(G)$ is the set of \emph{real antisymmetric} matrices supported on $G$.
 \end{itemize}

 For $h \in \cH(G)$, the eigenvalues 
$\lambda_1(h) \le \lambda_2(h) \le \cdots \le \lambda_n(h)$ are real and ordered.  
Hermitian matrices $h,h' \in \cH(G)$ are said to be \emph{gauge
  equivalent} (notation $h\sim h'$, cf.~\S \ref{subsec-gauge-invariance}) 
  if for some $\theta \in \RR^V$,
\begin{equation}
  \label{eq:gauge_equiv_def}
  h' = e^{i\theta} h e^{-i\theta},
  \qquad
  \text{where }e^{i\theta} = \diag(e^{i\theta_1}, e^{i\theta_2},\cdots, e^{i\theta_n}).
\end{equation}

The torus of \emph{magnetic perturbations} (cf.~\S \ref{subsec-magnetic-perturbation}) of a real symmetric matrix $h\in\cS(G)$, assumed to be \emph{strictly}\footnote{Being strictly supported on $G$ will be our standing assumption on $h\in \cS(G)$ throughout this paper.} supported on $G$, is the set of gauge-equivalence classes of Hermitian matrices,
\begin{equation}
  \label{eq:Mh_def}
  \tcM_h := \{h_{\alpha}\in\cH(G) \colon
  (h_{\alpha})_{rs} := e^{i\alpha_{rs}} h_{rs}, \ 
  \alpha\in\cA(G)\} / \sim.
\end{equation}
Of particular importance are the equivalence classes in $\tcM_h$ containing the {\em signings} of $h$: real symmetric matrices $h'$ with
$h'_{rr}=h_{rr}$ and $h'_{rs}=\pm h_{rs}$ for all $r\neq s$.

It is well known that $\tcM_h$ is
diffeomorphic to a torus\footnote{In \S\ref{subsec-gauge-invariance} below (see also \cite[Lem.~2.1]{LieLos_dmj93}) this torus is identified as $H^1(G, \RR) / H^1(G, 2\pi \ZZ)$, which is sometimes known as the Jacobi torus \cite{KotSun_aam00} of the graph $G$.} of dimension $\beta=|E|-n+1$.   
Because $h$ an $h'\sim h$ are unitarily equivalent, their spectra are identical, and the $k$-th
eigenvalue function on $\cH(G)$ descends to a function on $\tcM_h$,
\begin{equation}
  \label{eqn-lambda}
  \lambda_k:\tcM_h \to \RR
\end{equation}
which may be considered to be a sort of generalized Morse function \cite{BerZel_im24}, which loses its smoothness only at points of eigenvalue multiplicity.
We are interested in the \emph{smooth critical
points} $h_\alpha$ of this function, with motivation coming from two applications:
\begin{enumerate}
\item Band structure of the spectrum of periodic graphs, reviewed in
  \S\ref{sec:abelian_cover};
\item Distribution of the number of sign changes of graph
  eigenvectors, reviewed in \S\ref{sec:intro_nodal}.
\end{enumerate}

\subsection{Overview of results}
\label{sec:results_overview}

To give a high-level summary of our results,
\begin{itemize}
\item We establish that every smooth critical point of $\lambda_k:\tcM_h\to\RR$ lies on a critical submanifold associated to one of the finitely many choices of the ``critical data'' computable from $h$.  The computation boils down to finding the eigenvalue decomposition of a principal minor of a signing of $h$.
\end{itemize}

To give more detail, let $h \in \cS(G)$ be generic in the sense of \S\ref{ass}.
Fix $k \ge 1$ and suppose that $h_{\alpha} \in \tcM_h$ is a critical point of $\lambda_k$.
Suppose $\lambda = \lambda_k(h_{\alpha})$ is a simple eigenvalue with eigenvector $\psi$.  Then 
\begin{enumerate}
\item[(i)] the subgraph $G_N$ that is induced from the vertices 
\[ V_N = \left\{ v \in V:  \psi(v) \ne 0\right\}\] 
is an {\em admissible support} in the sense of \S \ref{subsec-admissible}.
The restriction $h_{\alpha}|G_N$ is gauge equivalent to a signing  $h_N$ of $ h|G_N$.  The restriction $\psi|V_N$ is gauge equivalent to an eigenvector
$\psi_N$ of $h_N$ with eigenvalue equal to $\lambda$.  The quadruple $(V_N, h_N, \psi_N, \lambda)$ is {\em critical data}
for $h$ in the sense of \S \ref{subsec-data}
\item[(ii)] The point $h_{\alpha}$ lies in the associated submanifold $F \subset \tcM_h$, see \eqref{eq:F}, consisting of the critical points that
share the same critical data $(V_N, h_N, \psi_N, \lambda)$.
It is diffeomorphic to a product \eqref{eqn-topology} of a torus and several 
{\em linkage spaces}.\footnote{Linkage spaces describe possible configurations of rigid links of prescribed length connected into a planar loop by joints; they are reviewed in Appendix \ref{sec:linkage}.}  It is stable (Theorem \ref{thm-stability})
with respect to perturbations of the underlying matrix $h \in \cS(G)$.
\item[(iii)] The Morse index of $\lambda_k$ at the point $h_{\alpha}$ is given by equation
\eqref{eq:MBindex} in terms of the number of sign changes of $\psi$ and the location of $\lambda$ in
the spectrum of $h_{\alpha}$ and its restrictions.
(The Morse index of $\lambda_k$ may fail to be constant on $F$, see Remark \ref{rem:jump_of_index}.)
\end{enumerate}

\subsection{Motivation 1: Band structure of the spectrum of the maximal abelian cover of \texorpdfstring{$G$}{G}.}
\label{sec:abelian_cover}

\newcommand{\wh}{{\widehat{h}}}
\newcommand{\wG}{{\widehat{G}}}

The maximal (or ``universal'') abelian cover $\wG$ of a finite
graph $G$ with the Betti number $\beta$ is the $\ZZ^\beta$-periodic graph
such that $\wG/\ZZ^\beta=G$.  The operator $h$ is lifted to a
periodic discrete Schr\"odinger operator $\wh$ acting on
$L^{2}\left(\wG\right)$, see
\cite{Sun_incol08,Sunada_TopologicalCrystallography,baez2025topologicalcrystals}.

Any operator can be decomposed into a direct sum (or integral) over the
irreducible representations of a (sub)group of its symmetries.  In the case when the symmetry group is $\ZZ^\beta$ (and the irreps are one-dimensional), this decomposition is known as the
Floquet--Bloch representation \cite{Kuc_bams16}.  As a result, we
obtain $\wh$ as a direct integral of $h_\alpha$ over a
$\beta$-dimensional torus of representations which one can recognize as the torus of magnetic perturbations $\tcM_h$.

The graph of the eigenvalues of $h_\alpha$ over $\tcM_h$ is known as
the ``dispersion relation'' or ``Bloch variety''.  It determines many
physical properties of the crystal $\wG$
\cite{AshcroftMermin_solid,Gieseker_fermi,Kuc_bams16} and the critical
points of the Bloch variety play a central role.  For instance, the
extrema of the $k$-th eigenvalue $\lambda_k(h_\alpha)$ determine the
edges of the band spectrum of $\wh$ \cite{DoKucSot_jmp20}.  In the
particular case of maximal abelian cover graphs, a complete
understanding of the critical points and their Morse data may open the
way to resolve the ``full spectrum'' conjecture of Higuchi and Shirai
\cite[Problem C]{HigShi_ymj99}, \cite[Conj.~2.4]{HigShi_dga04}: the
spectrum of $\wh$ has no gaps.

It is interesting to compare our results with a recent work of Faust and Sottile
\cite{FauSot_jst24} which used Bernstein--Kushnirenko theorem to
derive a bound on the total number of \emph{isolated} critical points
in the dispersion relation.  In the setting of maximal abelian cover
graphs, our results show an abundance of critical points that are
\emph{not isolated} but lie on critical submanifolds, which are fully
classified in Theorem~\ref{thm:main}.

\subsection{Motivation 2: Nodal count and nodal distribution}\label{sec:intro_nodal}

Suppose $\lambda_k(h)$ is a simple eigenvalue of an $h\in \cS(G)$ with
eigenvector $\psi$ which is chosen to be real.  Its {\em nodal count}
is
\begin{equation}
  \label{eqn-nodal-def}
  \nodcnt(h,k) = \left| \left\{(rs)\in E:\ \psi_r h_{rs} \psi_s > 0
    \right\} \right|.
\end{equation}
When $h$ is the Laplace operator or a discrete Schr\"odinger operator then $h_{rs} < 0$ and the nodal count is the number of
edges on which $\psi$ changes sign\footnote{In this work, the number of sign changes plays the principal role.  A related question of estimating the number of nodal domains is also well-studied \cite{DavGlaLeySta_laa01,Biy_laa03,XuYau_jc12,GeLiu_cvpde23}.}.

If the graph $G$ is a tree then $\nodcnt(h,k) = k-1$ \cite{Fie_cmj75,Ber_cmp08}, provided $\psi$ does not vanish on any vertices.  For a
general graph the difference
\begin{equation}
  \label{eq:nodal_surplus_def}
  \sigma(h,k) := \nodcnt(h,k) - (k-1)  
\end{equation}
is called
the {\em nodal surplus} and satisfies
\cite{Ber_cmp08,Ber_apde13} 
\[ 0 \le \sigma(h,k) \le \beta.\] 

A {\em signing} $h'$ of $h$ is a real symmetric matrix with
$h'_{rr}=h_{rr}$ and $h'_{rs}=\pm h_{rs}$ for all $r\neq s$.  The set
of equivalence classes containing a signing of $h$ consists exactly of
those points in $\tcM_h$ that are fixed under complex conjugation:
\begin{equation}
  \label{eq:signings_def} 
  \Sig_h := \{[h']\in \tcM_h \colon \text{$h'$ is a signing of
    $h$}\}=\{[h_{\alpha}]\in \tcM_h \colon
  \Big[\overline{h_{\alpha}} \Big] = [h_\alpha]\}.
\end{equation}
There are $2^{|E|}$ signings of $h$ and $2^{\beta}$ elements in $\Sig_h$.  
The spectra of signings of graph Laplacians has interesting connections to celebrated problems: for example, optimization of eigenvalues over all possible signings was used to construct Ramanujan graphs \cite{BilLin_comb06,MarSpi_am15a}
and later led to the resolution \cite{MarSpi_am15b} of the Kadison--Singer problem; maximization of eigenvalue multiplicity leads to optimal spherical two-distance sets \cite{JiaEtAl_c23} and, in particular, to maximization of equiangular line configurations \cite{JiaEtAl_am21}.

The {\em nodal surplus distribution} over all signings is the random variable
$\sigma(h',k)$ where $h'\in \Sig_h$ is a random signing of $h$.  It has been conjectured\footnote{See \Cref{sec-experiments} for history, evidence and a precise formulation of the conjecture.} that \emph{the nodal surplus distribution, once normalized, converges to standard normal} for any sequence of graphs $G$ with $\beta(G)\to\infty$.  The convergence is expected to be uniform in the choice of the graphs, choice of the matrices supported on the graphs and choice of the eigenvalue label $k$ (which may change from one graph to the next).  

 An analytical approach to the nodal surplus is provided by the
\emph{nodal magnetic theorem} of Berkolaiko \cite {Ber_apde13} and
Colin de Verdi\`ere \cite{Col_apde13}; it builds on the fact that the
signings are embedded in the torus of magnetic perturbations.
Suppose $\lambda_k$ is a simple eigenvalue of a signing $h'\in\Sig_h$ with
nowhere vanishing eigenvector.  A immediate consequence of the complex
conjugation symmetry of the points in $\Sig_h$ is that $h'$ is a
critical point of $\lambda_k$.  The nodal magnetic theorem \cite {Ber_apde13,Col_apde13}
states that this is a nondegenerate critical point and its Morse index is exactly the
nodal surplus $\sigma(h',k)$.

\begin{figure}
    \centering
    \includegraphics[width=1\textwidth]{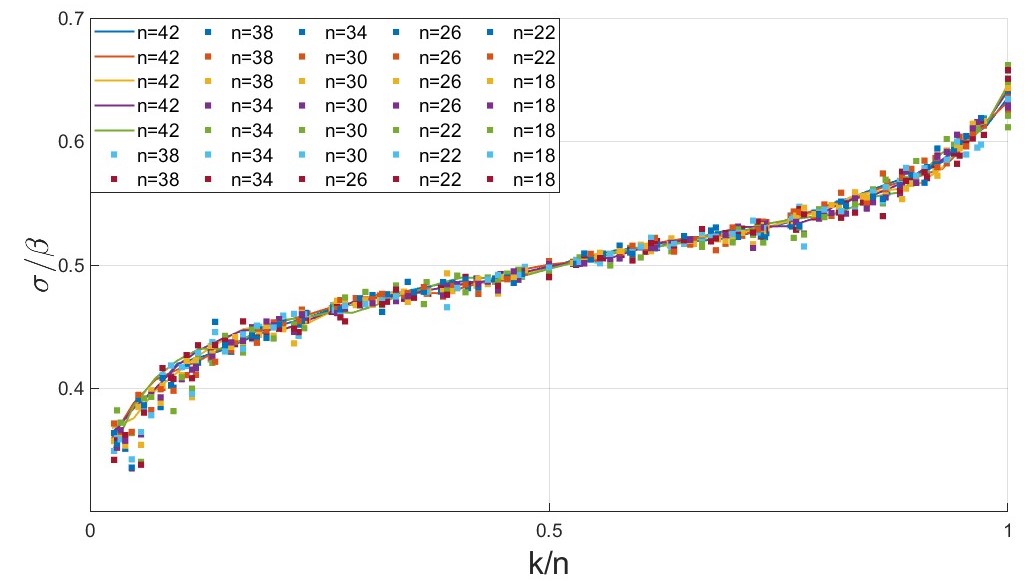}
    \caption{Plot of $\mathbb{E}(\sigma(h',k))/\beta(G)$, computed for
      a given $k$ and a random signing $h'$, against $k/n$. The matrix
      $h$ is the Laplacian of $G$ plus diagonal potential, and $G$
      ranges over 35 randomly chosen 3-regular graphs of $n$ vertices,
      see \Cref{sec-experiments-3reg}. } 
    \label{fig:Mean_surplus_curves4}
\end{figure}

One might therefore hope\footnote{Following a suggestion by Ram Band.}
to understand the distribution of nodal surplus over different
signings by analyzing the Morse theory of the eigenvalue map
\eqref{eqn-lambda}.  If, hypothetically, there were no critical points
of $\lambda_k$ on $\tcM_h$ other than the signings, it would
immediately imply that the nodal surplus distribution was
binomial\footnote{Morse inequalities \cite{Milnor_MorseTheory} for the
  torus $\TT^\beta$ state that there are at least $\binom{\beta}{s}$
  critical points of Morse index $s$, with all inequalities becoming
  sharp if the total number of critical points is $2^\beta$.}, with
the mean $\beta/2$.  Unfortunately, this Morse theoretic approach is
hindered by two phenomena, (a) critical points
$[h_\alpha] \not\in \Sig_h$ such that $\lambda_k(h_{\alpha})$ is a
simple eigenvalue and consequently $\lambda_k$ is smooth, and (b)
points $[h_\alpha] \in \tcM_h$ for which the eigenvalue $\lambda_k$
has multiplicity $\ge 2$ and $\lambda_k$ is, in general, non-smooth.
That these phenomena make a substantial contribution can be seen, for
example, from the experimental deviation of the mean of $\sigma(h',k)$ for
a random signing $h'\in\Sig_h$ from the binomial prediction $\beta/2$,
see \Cref{fig:Mean_surplus_curves4}.

In this paper we analyze the smooth critical points --- the points of
type (a) --- which turn out not to be isolated, in general.  Our
results can be viewed as an algorithm for computing all the critical
submanifolds and their Morse data by ranging over the finite
collection of signings of $h$ and admissible supports in $G$.  In
\S\ref{sec-experiments} we implement this algorithm in a numerical
investigation of $3$-regular graphs.  We discover that the symmetry
points $\Sig_h$ are but a tiny proportion of all critical point.
Remarkably, the conjecture for the universal behavior of the number of
sign changes of the eigenvectors (or, equivalently, the Morse index of
a random point in $\Sig_h$), still appears to hold.

Finally, our results also show that in the case of graphs with
disjoint cycles, there are no simple critical points except for
$\Sig_h$, opening a path to a simple proof of binomial nodal
distribution in this family of graphs, once the points of type (b) are
taken care of, for instance by using the stratified Morse theory developed
for the eigenvalues of generic self-adjoint matrix families in
\cite{BerZel_im24}.


\section{Further notation and definitions}

\subsection{Admissible support}
\label{subsec-admissible}

If $G=(V,E)$ is a simple connected graph then a subset
$V_N \subset V$ is called an {\em admissible support} if
\vspace{-5pt}
\begin{enumerate}
\item the subgraph $G_N$ induced by $V_N$ is connected and
\item each vertex $r \notin V_N$ has either no neighbors in $V_N$ or at least three neighbors in $V_N$.
\end{enumerate}
Examples of vertex subsets which are admissible and subsets which are
not are given in Figure \ref{fig:admissible}.
It will be established in Theorem~\ref{thm:main} that the support of a
critical eigenvector $\psi$ (i.e. the set of vertices where $\psi$ is
non-zero, hence the subscript $N$) is admissible.

Within the complement of $V_N$ we will distinguish vertices that are
connected to $V_N$ and the vertices that are not.
The {\em boundary vertices} and the {\em residual vertices} are
\begin{align}
  \label{eq:VZN_def}
  V_{ZN} &= \left\{ r \notin V_N: \exists s \in V_N \text{ with }\ (rs) \in E \right\},\\
  \label{eq:VZZ_def}
  V_{ZZ} &= V \setminus \big(V_N \cup V_{ZN}\big)
  = \left\{ r \notin V_N: (rs) \in E \Rightarrow s \notin V_N \right\}.
\end{align}
In the context of a vector $\psi$, $V_{ZN}$ is the set of zero
vertices (of $\psi$) that are connected to some non-zero vertices and
$V_{ZZ}$ are the zero vertices that are only connected to zero
vertices (see Figure~\ref{fig-tree} for an illustration).  We have the
disjoint decomposition
\begin{equation}
  \label{eq:V_decomp}
  V = V_N \sqcup V_{ZN} \sqcup V_{ZZ}. 
\end{equation}

\begin{figure}
    \centering
    \includegraphics[width=0.75\linewidth]{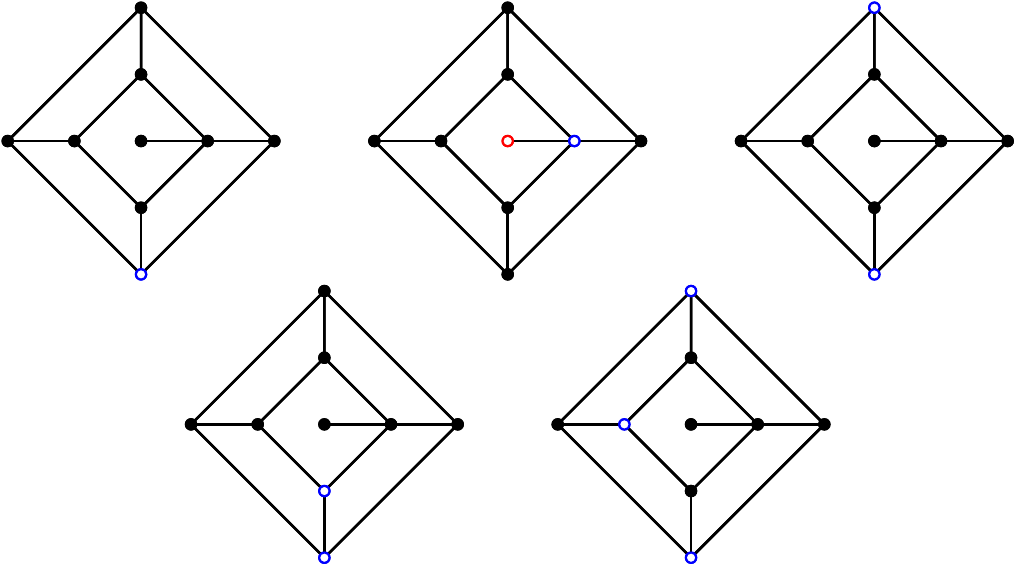}
    \caption{Top row: examples of admissible support (vertices in $V_N$ are solid black circles, vertices in $V_{ZN}$ and $V_{ZZ}$ are empty blue and red vertices, correspondingly).  Bottom row: examples of $V_N$ that are inadmissible (left: boundary degree is too small; right: the subgraph induced by $V_N$ is disconnected).}
    \label{fig:admissible}
\end{figure}

We also define the corresponding edge sets
\begin{align}
  \label{eq:edge_types}
  E_{NN} &= \left\{ (rs) \in E:\  r,s \in V_N\right\}, \\
  E_{ZZ} &= \left\{ (rs) \in E:\  r,s \not\in V_N \right\}, \\
  E_{ZN} &= \left\{ (rs)\in E: r \notin V_N, s \in V_N \right\},
\end{align}
with an example shown in \Cref{fig:VE_partitions}.

\usetikzlibrary{decorations.pathreplacing}
\begin{figure}[H]
\begin{tikzpicture}[scale = 1.5]

\foreach \x in {1,...,6}
  {\node [inner sep = 0pt](N\x) at (\x,4){};}
\node[inner sep = 0pt] (N0) at (4,5){};

\foreach \x in {1,2,3}
  {\node[inner sep = 0pt](ZN\x) at (2*\x,3){};}

\foreach \x in {1,...,3}
{\node[inner sep = 0pt](ZZ\x) at (2*\x,2){};}

\draw[thick] (N1) to [out = 45, in = 135] (N3);
\draw[thick] (N1) to [out = 90, in = 180] (N0);
\draw[thick] (N1) -- (N5);
\foreach \x in {3,...,6}
  {\draw[thick] (N0) -- (N\x);}

\foreach \x in {2,...,5}
  {\draw[line width=1.5, color=cyan] (ZN2) -- (N\x);}
\foreach \x in {1,...,3}
  {\draw[line width=1.5, color = cyan] (ZN1) -- (N\x);}
\foreach \x in {4,...,6}
  {\draw[line width=1.5, color = cyan](ZN3) -- (N\x);}

\draw[line width=1.5, color=red] (ZN1) -- (ZN2)--(ZZ2)--(ZZ1);
\draw[line width=1.5, color=red] (ZN3) -- (ZZ3);

\foreach \x in {0,...,6}
{\draw[fill] (N\x) circle [radius=.1];}

\foreach \x in {1,...,3}
{\draw[fill, color=cyan] (ZN\x) circle [radius = .1];
  \draw[fill, color=white] (ZN\x) circle [radius =.06];}

\foreach \x in {1,...,3}
{\draw[fill, color=red] (ZZ\x) circle [radius = .1];
\draw[fill, color=white] (2*\x,2) circle [radius = .06];}

  \draw[decorate,decoration={brace,mirror,  amplitude = 13}]    (7,3.7)--(7,5.3) node [midway, xshift = 1cm] {$V_N$}    ;
  \draw[decorate,decoration={brace,mirror,  amplitude = 7}]    (7,2.75)--(7,3.25) node [midway, xshift = 1cm, cyan] {$V_{ZN}$}    ;
  \draw[decorate,decoration={brace,mirror,  amplitude = 7}]    (7,1.75)--(7,2.25) node [midway, xshift = 1cm, red] {$V_{ZZ}$}    ;

\end{tikzpicture}
\caption{An example of a graph with an admissible support $V_N$ shown
  as black vertices (the vertices in $V_{ZN}$ and $V_{ZZ}$ are cyan
  and red empty circles).  The edges $E_N$, $E_{ZN}$, and $E_{ZZ}$ are shown in
  black, cyan, and red, respectively.}
\label{fig:VE_partitions}
\end{figure}
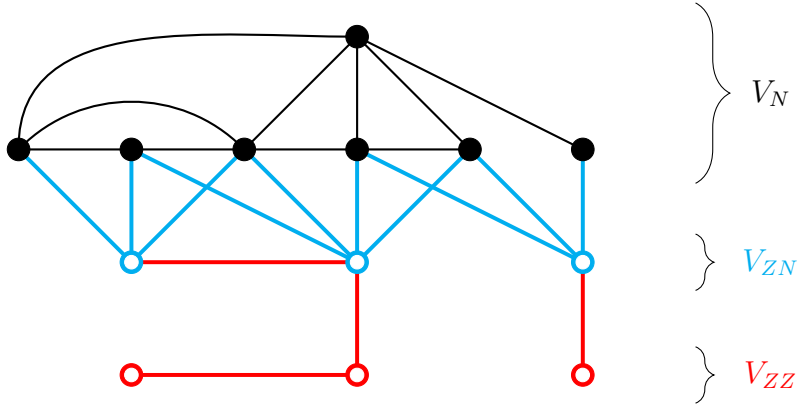

We will denote by $G_N$, $G_{ZN}$ and $G_{ZZ}$ the sugraphs induced by
the vertex sets $V_N$, $V_{ZN}$ and $V_{ZZ}$.  Note that $E_{NN}$ is
the edge set of the subgraph $G_N$, but the edge sets of $G_{ZN}$ and
$G_{ZZ}$ are \emph{not} $E_{ZN}$ and $E_{ZZ}$.  In fact, they are
subsets of $E_{ZZ}$ and
\begin{equation}
  \label{eq:EzzVzz}
  |E_{ZZ}| - |V_{ZZ}| \geq 0,
\end{equation}
which will be of importance later.

The notation $h|G_N$ will be used for the subblock of the matrix
$h \in \cH(G)$, consisting of the elements of $h$ with indices in
$V_N$.  Similarly, $h|G_{ZZ}$ denotes the restriction of $h$ to the
subgraph $G_{ZZ}$ that is induced by the vertices $V_{ZZ}$.

\subsection{Genericity conditions}
\label{ass}
In the results below, we will use
the following conditions on $h\in\cS(G)$.
\begin{enumerate}[label=(\roman*)]
\item \label{strict}
  The matrix $h$ is strictly supported on $G$.
\item \label{item-subset} For any induced subgraph $G_J=(J, E_J)$ of
  $G$ (including the case $G_{J}=G$), and any signing $h'\in \Sig_{h|G_J}$, all eigenvalues $h'$ are
  simple, and all eigenvectors are nowhere-vanishing on exactly one
  connected component of $G_J$.
\end{enumerate}  

\begin{prop}
  \label{prop-genericity}
  The set of matrices in $\cS(G)$ that fail to satisfy Assumptions
  \ref{ass} is a semi-algebraic subvariety of codimension $\ge 1$.
  Consequently, the set of matrices that satisfy assumptions \ref{ass}
  is open and dense in $\cS(G)$.
\end{prop}

The proof appears in \S\ref{sec:generic_consequences}.

\section{Statement of results}

\subsection{Critical data} \label{subsec-data}
Let $h \in \cS(G)$.  {\em Critical data} for $h$ is a quadruple
$(V_N, h_N, \psi_N, \lambda)$, where
\begin{itemize}
\item $V_N \subset V$ is an admissible support with induced graph
  $G_N$, see \S~\ref{subsec-admissible},
\item $h_N \in \Sig_{h|G_N}$ is a signing of the restriction 
  $h|G_N$ of $h$ to  $G_N$,
\item $\psi_N \in \CC^{V_N}$ is an eigenvector of $h_N$ with
  the eigenvalue $\lambda = \lambda_{k_N}(h_N)$.
\end{itemize}  To each choice of critical data we  associate a (possibly empty) subset
of $\tcM_h$, 
\begin{equation}
  \label{eq:F}
  F = F(V_N, h_N, \psi_N, \lambda)
  =\{[h_{\alpha}]\in\tcM_h :\ h_\alpha |G_N = h_N \text{ and } 
  h_\alpha \psi = \lambda \psi \},
\end{equation}
where $\psi$ is the extension of $\psi_N$ to the rest of the graph
$G$ by 0.

\begin{thm}
  \label{thm:main}
  Let $h \in \cS(G)$ satisfy the genericity hypotheses \S\ref{ass}.
  Let $(V_N, h_N, \psi_N, \lambda)$ be critical data for $h$ with
  associated $F = F(V_N,h_N,\psi_N,\lambda)$.  Then:
  \begin{enumerate}
  \item \label{item:Ftopology} The set $F$ is a smooth submanifold of
    $\tcM_h$ diffeomorphic to a product of a torus with one or more
    linkage manifolds,
    \begin{equation}
      \label{eqn-topology}
      F \cong \TT^{|E_{ZZ}| - |V_{ZZ}|} \times \prod_{r \in
        V_{ZN}}M(b_r),
    \end{equation}
    where each $M(b_r)$, $r\in V_{ZN}$, is the linkage manifold (see
    Appendix \ref{sec:linkage}) associated to the vector
    \begin{equation}
      \label{eq:br_def}
      b_r = \left\{ \big| \psi_s h_{rs} \big| \colon (rs) \in E_{ZN}\right\}.
    \end{equation}
    The manifold $F$ is non-empty if and only if
    \begin{equation}
      \label{eq:nontriv_F}
      2 \max_{s\in V_{N}} |h_{rs}\psi_{s}|
      \le 
      \sum_{s\in V_{N}}|h_{rs}\psi_{s}|,
      \qquad \text{for each } r\in V_{ZN},
    \end{equation}
    in which case $F$ has $\prod_{r \in V_{ZN}}\beta_0(b_r)$
    (see \Cref{thm-linkage}\eqref{item:connected_components})
    connected components and 
    \begin{align}
      \label{eq:MBdimension}
      &\dim(F) =
        |E_{ZN}| - 3|V_{ZN}| + |E_{ZZ}| - |V_{ZZ}|, \\
      &\codim(F) = \beta(G_N) + 2 |V_{ZN}|.
    \end{align}
  \item \label{item:Fcritical} Let $[h_{\alpha}] \in F$ and let $k$ be
    the label of the eigenvalue $\lambda$ in the spectrum of
    $h_\alpha$, cf.~\eqref{eq:F}.  If $\lambda_k(h_\alpha)$ is
    simple, then $[h_\alpha]$ is a critical point of $\lambda_k$ on
    $\tcM_h$.  Conversely, each simple critical point of $\lambda_k$
    belongs to one of the submanifolds $F(V_N, h_N, \psi_N, \lambda)$.
  \item \label{item:MBindex} If $\lambda_k(h_\alpha)$ introduced in
    \eqref{item:Fcritical} is simple, and
    $\lambda =\lambda_k(h_\alpha) \not\in
    \spec\big(h_\alpha|G_{ZZ}\big)$, then the critical point
    $h_\alpha$ is non-degenerate in the directions normal to $F$, with
    the Morse index
    \begin{equation}
      \label{eq:MBindex}
      \ind(\lambda_k, h_\alpha)
      = \sigma(h_N, k_N) +2|V_{ZN}|-2(k-k_{ZZ}-k_N),
    \end{equation}
    where $k_N$ is the label of the eigenvalue $\lambda$ in the
    spectrum of $h_N$ and $k_{ZZ}$ is the number of eigenvalues of
    $h_\alpha|G_{ZZ}$ less than $\lambda$.
  \end{enumerate}
\end{thm}

\begin{rem}
  \label{rem:minimum}
  Equation \eqref{eq:MBindex} provides sufficient and necessary
  conditions for $[h_{\alpha}]\in F$ being a local extremum of
  $\lambda_{k}$ (and thus a candidate for the role of a band edge,
  cf.~\S\ref{sec:abelian_cover}):
  \begin{enumerate}
  \item $[h_{\alpha}]$ is a local minimum if and only if both
    $\sigma(h_{N},k_{N})=0$ and $k=|V_{ZN}|+k_{ZZ}+k_N$.
  \item $[h_{\alpha}]$ is a local maximum if and only if both
    $\sigma(h_{N},k_{N})=\beta(G_{N})$ and $k=k_{ZZ}+k_N $.
  \end{enumerate}
  This holds because the terms
  $\sigma(h_N, k_N)$ and $k-k_{ZZ}-k_N$ are non-negative
  and bounded by $\beta(G_{N})$ and $|V_{ZN}|$, respectively.
  Note that the necessary conditions $\sigma(h_{N},k_{N})=0$ and
  $\sigma(h_{N},k_{N})=\beta(G_{N})$ can be verified before computing
  $h_{\alpha}$.
\end{rem}

\begin{rem}
  \label{rem:criticality_condition}
  Recall that $h_{\alpha}$ is a critical point of a simple eigenvalue
  $\lambda_k:\tcM_h \to \RR$ if and only if the following {\em
    criticality condition} holds:
  \begin{equation}
    \label{eqn-criticality}
    (h_{\alpha})_{rs} \bar\psi_r \psi_s \in \RR \ \text{ for all } r
    \sim s,
  \end{equation}
  see \cite[Lem.~A.2]{BanBerWey_jmp15} or \cite[Thm.~3.2(3)]{AloGor_jst23}.
  Therefore, at a critical point $h_\alpha$, the nodal count
  $\nodcnt(h_{\alpha},k)$ of \eqref{eqn-nodal-def} and the nodal
  surplus $\sigma(h_\alpha,k)$ of \eqref{eq:nodal_surplus_def} are
  well-defined.
\end{rem}

\begin{rem}
  \label{rem:signings}
  Since $V_N=V$ is a valid admissible support, in which case $h_N$ is
  a signing of $h$ and $F$ is a single point, our results include the
  original nodal magnetic theorem of Berkolaiko and Colin de
  Verdi\`ere \cite{Ber_apde13,Col_apde13}: any signing $h'$ of $h$ is a
  non-degenerate critical point of $\lambda_k$ on $\tcM_h$ with the
  Morse index equal to the nodal surplus $\sigma(h',k)$.
\end{rem}

\begin{cor}
  \label{cor:disjoint_cycles}
  If the graph $G$ has disjoint cycles (see \cite{AloGor_prep24}), then
  for a generic $h\in\cS(G)$, the only simple critical points are its
  signings $\Sig_h$.
\end{cor}

\begin{proof}
  The only admissible support on a graph with disjoint cycles is $V$.
\end{proof}

\begin{rem}
  \label{rem:jump_of_index}
  It cannot be guaranteed that the Morse index is constant along the
  manifold $F$.  A change in the index can happen due to failure of
  the conditions in part~\eqref{item:MBindex} of
  Theorem~\ref{thm:main}, namely:
  \begin{itemize}
  \item The index $k$ of the eigenvalue $\lambda$ in the spectrum of
    $h_\alpha$ may change along $F$.  At the change-points, $\lambda$
    becomes a multiple eigenvalue and $\lambda_k$ is in general not
    differentiable.
  \item The value of $k_{ZZ}$ may also change; at the change-points,
    $\lambda$ will be an eigenvalue of $h_\lambda|G_{ZZ}$.
  \end{itemize}
  
  In \Cref{sec:examples} we will demonstrate examples of both occurrences.
  Notably, equation~\eqref{eq:MBindex} shows that the Morse index
  can only change along $F$ by multiples of 2.
\end{rem}

Parts \eqref{item:Ftopology}-\eqref{item:Fcritical} of
\Cref{thm:main} are proven in \S\ref{sec:constructingF}.  The
computation of the index of the Hessian (which is the most substantial
piece of this paper) is done in \S\ref{sec:theHessian}.

In the next theorem we show that for every graph $G$, every admissible
support $V_N$ in $G$, and every $k$, we can construct matrices
$h\in\cS(G)$ for which the $k$-th eigenvalue has a critical point
$h_\alpha$, with eigenvector supported on $V_N$.

\begin{thm}\label{thm-existence}
  Let $G=(V,E)$ be a finite simple connected graph.  Suppose
  $V_N \subset V$ is an admissible support and
  $V = V_N \sqcup V_{ZN} \sqcup V_{ZZ}$ is the decomposition
  introduced in \eqref{eq:V_decomp}.  Let $G_N$, $G_{ZN}$
  and $G_{ZZ}$ be the subgraphs of $G$ induced by $V_N$, $V_{ZN}$ and
  $V_{\ZZ}$ respectively.

  Let $h_N \in \cS(G_N)$ have a simple eigenvalue
  $\lambda = \lambda_{k_N}(h_N)$ with (real) eigenvector $\psi_N$
  which is non-zero on all of $V_N$.  Let $h_{ZN} \in \cS(G_{ZN})$ and
  $h_{ZZ} \in \cS(G_{ZZ})$ be arbitrary matrices such that $\lambda$
  is \emph{not an eigenvalue} of $h_{ZN}$ and $h_{ZZ}$; denote
  \begin{equation}
    \label{eq:k_{ZZ}_def}
    k_{ZN} := \left| \left\{\widetilde\lambda \in
        \spec(h_{ZN}) \colon \widetilde\lambda < \lambda \right\}
    \right|,
    \qquad
    k_{ZZ} := \left| \left\{\widetilde\lambda \in
        \spec(h_{ZZ}) \colon \widetilde\lambda < \lambda \right\}
    \right|.
  \end{equation}

  Then, for any $\epsilon>0$, there exists a generic $h \in \cS(G)$,
  cf.\ \S\ref{ass}), and $h_\alpha \in \tcM_h$ such that
  \begin{enumerate}
  \item $h$ has (approximately) prescribed blocks
    \begin{equation}
      \label{eq:h_blocks_approx}
      \| h|G_N - h_N \| <\epsilon, \qquad
      \|h|G_{ZN} - h_{ZN} \| < \epsilon, \qquad
      \|h|G_{ZZ} - h_{ZZ}\|<\epsilon.   
    \end{equation}
  \item $h_{\alpha}$ is a critical point of
    $\lambda_{k_N+k_{ZN}+k_{ZZ}}$, which is a simple eigenvalue with
    the corresponding eigenvector supported on $V_N$.
  \end{enumerate}
\end{thm}

We remark that by choosing $h_N$, $h_{ZN}$ and $h_{ZZ}$ appropriately
we can get a critical point for any eigenvalue $\lambda_{k'}$ of $h_\alpha$.
The proof of Theorem~\ref{thm-existence} appears in \S \ref{proof-existence}.

The final main result shows that critical manifolds $F$ are stable under
perturbations of $h$.
\begin{thm}
  \label{thm-stability}
  Fix $h \in \cS(G)$, generic in the sense of \S\ref{ass}, and let
  $F=F(V_N, h_N, \psi_N, \lambda)$ be a non-empty critical submanifold
  constructed in accordance with Theorem~\ref{thm:main}.

  Then for any $\epsilon>0$, there is a neighborhood
  $U_h \subset \cS(G)$ such that every $\tcM_{h'}$, $h' \in U_h$, has
  a critical manifold $F' = F'(V_N, h_N', \psi_N', \lambda')$
  diffeomorphic to $F$, based on the same admissible set $V_N$, and
  constructed in accordance with Theorem~\ref{thm:main} from a
  critical data $(V_N, h_N', \psi_N', \lambda')$ close to
  $(V_N, h_N, \psi_N, \lambda)$, namely
  \begin{equation}
    \label{eq:eig_stability}
    \|h_N' - h_N\| < \epsilon,
    \qquad
    \left\| \psi_N' - \psi_N \right\| < \epsilon,
    \qquad
    |\lambda'-\lambda| < \epsilon.
  \end{equation}
\end{thm}

\begin{proof}
  Assume, without loss of generality, that $h_N = h |G_N$.  Viewing
  the matrix $h_N' = h'|G_N$ as a perturbation of $h|G_N$, the
  simple\footnote{By the genericity assumption.}  eigenvalue $\lambda$
  has a continuation with a continuous eigenvector, yielding
  inequalities \eqref{eq:eig_stability} in a small enough neighborhood
  $U_h$ of $h$.  The vectors $b_r$, $r\in V_{ZN}$, are also
  continuous, which allows us to cite
  \Cref{thm-linkage}\eqref{item:continuous} to establish that $F'$ is
  diffeomorphic to $F$.
\end{proof}

\section{Genericity conditions and their consequences}
\label{sec:generic_consequences}

To prove \Cref{prop-genericity}, we first formalize the notion of
failing Assumption~\ref{ass}.
\begin{defn}[$\bad(G)$]
  Given a graph $G$ of $n>1$ vertices, not necessarily connected,
  denote by $\bad(G)\subset\cS(G)$ the set of all matrices
  $h\in\cS(G)$ that have a signing $h'$ satisfying one of the
  following conditions:
  \begin{enumerate}
  \item $(h')_{ij}=0$ for some edge $(ij)\in E(G)$, i.e., $h'$ is not
    \emph{strictly} supported on $G$, or
  \item $h'$ has a multiple eigenvalue, or
  \item $h'$ has an eigenvector that has a vanishing entry on every
    connected component of $G$. 
  \end{enumerate}
  For convenience, set $\bad(G)=\emptyset$ when $n=1$.
\end{defn}

We can now reformulate \Cref{prop-genericity} as follows.
\begin{prop}
  \label{prop:genericity_again}
  The subset of matrices $h\in\cS(G)$ such that $h|G' \in \bad(G')$ for
  some induced subgraph $G' \subset G$ is semi-algebraic of positive
  codimension.
\end{prop}

\begin{lem}
  \label{lem:subgraph_up}
  If $G' \subset G$ is a subgraph and $X \subset \cS(G')$ is
  semi-algebraic of positive codimension, then
  \begin{equation}
    \label{eq:subgraph_up}
    \{ h \in \cS(G) \colon h | G' \in X\}
  \end{equation}
  is a semi-algebraic subset of $\cS(G)$ of positive codimension.
\end{lem}

\begin{proof}
  The set in \eqref{eq:subgraph_up} is the preimage of the projection
  $h \mapsto h|G'$ and has the structure of a direct product of $X$
  with $\RR^m$ for some $m$.
\end{proof}

\begin{prop}[cf.\ {\cite[Prop.~6.3]{AloGor_prep24}}]
  For any graph $G$ the set $\bad(G)$ is a semi-algebraic set of positive codimension. 
\end{prop}

\begin{proof}
  This claim was established in \cite[Prop.~6.3]{AloGor_prep24} under
  the assumption that $G$ is connected.

  Assume now that $G$ is not connected and denote its connected components
  by $G_{j}$, $j=1,\ldots,m$. Notice that if $h\in\bad(G)$ then
  either $h|G_{j}\in\bad(G_{j})$ for some $j$ or $h$ has a multiple
  eigenvalue. Recall that there is a polynomial on $\cS(G)$,
  called \emph{discriminant}, that vanishes if and only if $h$
  has multiple eigenvalues; let $Z(G)$ be its zero set. Then,
  \begin{equation}
    \label{eq:big_bad}
    \bad(G) = Z(G) \cup \bigcup_{j=1}^{m} \{h\in\cS(G)\colon h|G_{j}\in\bad(G_{j})\}.    
  \end{equation}
  The set $Z(G)$ is algebraic, and since a diagonal matrix with
  distinct entries is not in $Z(G)$, we conclude that $Z(G)$ has
  positive codimension.  Each $\bad(G_{j})$ is semi-algebraic of
  positive codimension by \cite[Prop.~6.3]{AloGor_prep24}.  By
  \Cref{lem:subgraph_up}, each of the sets in the union is also
  semi-algebraic of positive codimension.  We conclude that their
  (finite) union $\bad(G)$ is semi-algebraic of positive codimension.
\end{proof}

\begin{proof}[Proof of \Cref{prop:genericity_again} and hence of \Cref{prop-genericity}]
  For any induced subgraph $G' \subset G$, $\bad(G')$ is
  semi-algebraic of positive codimension in $\cS(G')$.  We now take
  the preimage of $\bad(G')$ by \Cref{lem:subgraph_up} and take the
  finite union over all induced subgraphs of $G$.
\end{proof}

Our genericity assumptions imply that certain planar linkages defined
from eigenvectors of subgraphs of $G$ are generic.

\begin{lem}
  \label{lem:generic_linkage}
  Let $h\in \mathcal S(G)$ satisfy Assumptions~\ref{ass} and let $V'$
  be a subset of vertices such that the induced subgraph $G'=(V',E')$
  is connected.  If $\psi$ is an eigenvector of $h|G'$, then, for
  every $r \not\in V'$ the vector
  $\left( \left| h_{rs}\psi_s \right| \right)_{s\in V',\, s\sim r}$ defines a generic
  linkage, namely (cf.~\eqref{eqn-generic-linkage})
  \begin{equation}
    \label{eq:evec_generic_linkage}
    \sum_{s \in V'} \epsilon_s h_{rs} \psi_s \neq 0,
    \qquad \text{for all } \epsilon_s = \pm1,
  \end{equation}
\end{lem}

\begin{proof}
  Suppose \eqref{eq:evec_generic_linkage} fails for some vertex $r$.
  Consider the subgraph $G''$ induced by the vertex set $V'
  \cup \{r\}$.  Let $h''$ be a signing of $h$
  obtained by modifying $h_{rs}'' = \epsilon_s h_{rs}$.  Then
  $(\psi, 0)^t$ is an eigenvector of $h'' | G''$ in contradiction
  to the non-vanishing assumption in \S\ref{ass}\eqref{item-subset}.
\end{proof}

\section{Magnetic perturbation and gauge equivalence}
\label{sec:mag_gauge}

\subsection{Magnetic perturbation}\label{subsec-magnetic-perturbation}
Recall the vector spaces $\cS(G)$, $\cA(G)$ and $\cH(G)$ defined in
\S\ref{sec:intro_nodal}.  The ordering of the vertices determines an orientation of
each edge but the set $E$ is viewed as the set of undirected edges,
i.e.\ we do not distinguish $rs \in E$ from $sr \in E$.
The space $\cA(G)$ acts on $\cH(G)$ by 
\begin{equation}
  \label{eq:alpha_star}
  (\alpha*h)_{rs} = e^{i \alpha_{rs}}h_{rs}.
\end{equation}
For a given $h$, the mapping
\begin{equation}
  \label{eq:action_alpha}
  \alpha \mapsto \alpha*h =: h_\alpha
\end{equation}
maps $\cA(G)$ onto the space
\begin{equation}
  \label{eq:space_cM}
  \cP_h := \{h' \in \cH(G) \colon
  |h'_{uv}| = |h_{uv}| \text{ for all }u\neq v
  \text{ and }h'_{uu} = h_{uu}
  \text{ for all }u\in V\}.
\end{equation}
Because the elements $\alpha_{rs}$ only play a role in
\eqref{eq:alpha_star} modulo $2\pi$, topologically $\cP_h$ is a torus.
The mapping $*h \colon \cA(G) \to \cP_h$ in \eqref{eq:action_alpha}
may be interpreted as the universal covering of $\cP_h$; since
$\cA(G)$ is linear, we may identify it with the tangent space (at any
point of $\cP_h$)
\[
  T_{h_{\alpha}}\cP_h \cong \cA(G).
\]

\subsection{Graph cohomology}\label{subsec-cohomology}
We may consider a vector
$\theta = (\theta_1,\theta_2,\cdots,\theta_n) \in \RR^V$ to be a
0-form on $G$, that is, a real valued function defined on the vertices of $G$.
We may  view $\cA(G)$ as the space of 1-forms on $G$.  The
coboundary operator $d=d_G : \RR^V \to \cA(G)$ is given by
\begin{equation}
  \label{eq:differential_def}
  (d\theta)_{rs} =
  \begin{cases}
    \theta_r - \theta_s & \text{ if } (rs) \in E \\
    0 & \text{ otherwise}.
  \end{cases}  
\end{equation}
Since $G$ is connected, the kernel of $d$ has dimension 1 and
$H^1(G,\RR) = \cA(G)/d\RR^V$ is a vector space of dimension
$\beta = |E|-|V| + 1$.

The space $C_1(G,\RR)$ of 1-chains consists of formal linear
combinations of edges $(rs)\in E$ with $r<s$.  The space $C_0(G,\RR)$
of $0$-chains consists of formal linear combinations of vertices.  The
boundary map $\partial:C_1(G,\RR) \to C_0(G,\RR)$ is defined via
$\partial (rs) = (s) - (r)$, extended by linearity.  The kernel of
$\partial$ consists of {\em cycles} and, since there are no 2-chains,
gives the first homology $H_1(G, \RR) = \ker(\partial)$.

The spaces $C^1(G,\RR)$ and $C_1(G,\RR)$ are dually paired with
\begin{equation}
  \label{eqn-integral}
  \langle \alpha, \xi \rangle = \int_{\xi} \alpha
  := \sum_{\substack{(rs)\in E\\r<s}} \alpha_{rs} \xi_{rs}
\end{equation}
where $\alpha \in \mathcal A(G)$ is a 1-form, and $\xi$ is the $1$-chain
\[
    \xi = \sum_{\substack{(rs)\in E\\r<s}}\xi_{rs}\, (rs) \in C_1(G,\RR).
\]
 Stokes' theorem $\int_{\xi}d_G\theta = \int_{\partial \xi}\theta$
  holds, so this pairing  induces a dual pairing between $H^1(G,\RR)$ and $H_1(G, \RR)$. If $\xi$ is a cycle, the integral
\eqref{eqn-integral} is called the {\em flux} of $\alpha$ through $\xi$.
   In particular $\alpha \in d_G \RR^V$ if and only if its flux
  through any cycle $\xi$  is zero.

\subsection{Gauge equivalence}\label{subsec-gauge-invariance}  If $\theta \in \RR^V$ and $\alpha \in \cA(G)$
then a direct calculation using \eqref{eq:alpha_star} gives
\begin{equation}
  \label{eq:dtheta_action}
  d\theta * h_{\alpha} = d\theta*\alpha*h =(d\theta+\alpha)*h = 
  e^{i\theta} h_{\alpha} e^{-i\theta}  
\end{equation}
where $e^{i\theta} = \diag(e^{i\theta_1}, e^{i\theta_2},\cdots, e^{i\theta_n})$.  
Thus $h_{\alpha}$ and $h_{\alpha'}$ are gauge equivalent (see \eqref{eq:gauge_equiv_def})
if and only if $\alpha' = \alpha + d\theta \mod 2\pi$ for
some $\theta \in \RR^V$.  In other words, $\alpha$ and $\alpha'$ are representatives of the same equivalence class in $\cA(G)/d\RR^V = H^1(G,\RR)$ modulo $H^1(G,2\pi\ZZ)$.
The space $\tcM_h$ of magnetic perturbations of
$h \in \cS(G)$ is the quotient of the space $\cP_h$
by gauge equivalence,
\begin{equation}
  \label{eq:M-quotient-P}
  \tcM_h = \cP_h / \sim,
\end{equation}
and is therefore diffeomorphic to the torus $H^1(G,\RR)/H^1(G,2\pi\ZZ)$ of
dimension $|E|-|V|+1 = \beta(G)$.  This torus is sometimes called the
``Jacobi torus'' in the literature \cite{KotSun_aam00}.

\subsection{Spanning trees}
\label{subsec-lift}

Choose a set of edges $E_G^{\f} \subset E$ whose complement is a
spanning tree of $G$. The set $\tcA(G)$ of antisymmetric matrices
supported on $E_G^\f$ is a (real) vector space of dimension
$|E_G^{\f}| = \beta(G)$.  This gives a canonical direct sum
decomposition (``Hodge-type'' decomposition)
\begin{equation}
  \label{eq:direct_sum_N}
  \cA(G) = \tcA(G) \oplus d\RR^V,
\end{equation}
and identifies $\tcA(G)$ as a section of the projection $\cA(G) \to
\cA(G) / d\RR^V$.  The mapping
\begin{equation}
  \label{eq:covering_Mh}
  \alpha \in \tcA(G)
  \quad \mapsto \quad
  [\alpha * h] \in \tcM_h
\end{equation}
may be identified with the universal covering of the torus $\tcM_h$,
which gives an identification of $\tcA(G)$ with the tangent space of
$\tcM_h$ at any point $[h_{\alpha}] \in \tcM_h$.  The choice of the
free edges $E_G^{\f}$ (or, equivalently, the choice of a spanning
tree) gives an embedding of $\tcM_h$ into
$\cP_h \subset \mathcal H(G)$: it is the image under $*h$ of a
universal torus $\tcA(G)\mod~2\pi$.

The objects defined in the above paragraphs are summarized in the
Figure \ref{fig:diagram}. The horizontal arrows (on the left) are covering
    maps, the downward arrows are the quotient maps, and the upward
    arrows are inclusions induced by a choice of the spanning tree.  Since all
    matrices in an equivalence class in $\tcM_h$ are unitarily
    equivalent, the eigenvalue function $\lambda_k$ on $\cP_h
    \subset \cH(G)$ descends to a function on $\tcM_h$.

\begin{figure}[H]
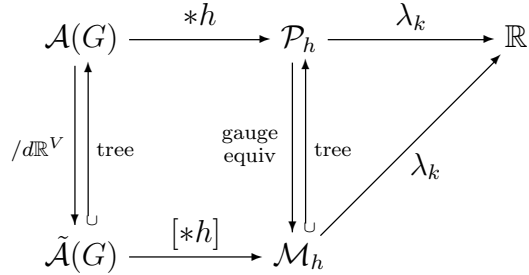

  
  \centering
  \begin{diagram}[size=3.5em,p=5pt]
    \cA(G) & \rTo^{*h} & \cP_h & \rTo^{\lambda_k} & \RR\\
    \dTo^{\scriptstyle{/d\RR^V}}\uInto_{\text{\tiny{tree}}} &&  \dTo^{\begin{smallmatrix}\textrm{\tiny{gauge}}\\
    \textrm{\tiny{equiv}}\end{smallmatrix}}\!\uInto_{\text{\tiny{tree}}}
    & \ruTo_{\lambda_k} &\\
    \tcA(G) & \rTo^{[*h]} & \tcM_h &&
  \end{diagram}
  \caption{A summary of the objects used in the paper and their
    relationships.}  \label{fig:diagram}
\end{figure}

\subsection{\texorpdfstring{$V_N$-adapted coordinates on $\cA(G)$ and
    $\tcM_h$}{Coordinates on the magnetic torus}}
\label{subsec:tree_choice}

Given a subset $V_N$, we have a partition of the edges $E$ into the
disjoint union
\begin{equation}
  \label{eq:edge_disjoint_union}
  E = E_{NN} \sqcup E_{ZN} \sqcup E_{ZZ}.
\end{equation}
Let $\cA_{N}$, $\cA_{ZN}$ and $\cA_{ZZ}$ denote the antisymmetric
matrices from $\cA(G)$ supported on the edge sets $E_{NN}$, $E_{ZN}$
and $E_{ZZ}$, correspondingly.  From \eqref{eq:edge_disjoint_union} we have
\begin{equation}
  \label{eqn-AG}
  \cA(G) = \cA_N \oplus \cA_{ZN} \oplus \cA_{ZZ}.
\end{equation}

Based on $V_N$, we now choose a spanning tree in a particular fashion.
This spanning tree in $G$ is constructed in three steps:
\begin{enumerate}
\item Choose a spanning tree in the subgraph $G_N = (V_N, E_{NN})$.
\item For each vertex  $v\in V_{ZN}$ choose exactly one edge $(v,w) \in E_{ZN}$.
\item Complete this to a spanning tree for $G$ by choosing edges in $E_{ZZ}$.
\end{enumerate}
A spanning tree chosen in accordance with this algorithm is
illustrated in Figure \ref{fig-tree}.

Denoting by $E_G^\f$ the edges in the complement of the chosen tree,
we define
\begin{equation}
  \label{eq:NZ_free_edges}
  E_N^\f = E_N \cap E_G^\f,
  \qquad
  E_{ZN}^\f = E_{ZN} \cap E_G^\f,
  \qquad
  E_{ZZ}^\f = E_{ZZ} \cap E_G^\f.
\end{equation}
We can now define $\tcA_{N}$, $\tcA_{ZN}$ and $\tcA_{ZZ}$ to be the
spaces of the antisymmetric
matrices supported on the edge sets $E_{NN}^\f$, $E_{ZN}^\f$
and $E_{ZZ}^\f$.  Since the ``free'' edges are the complement of a
spanning tree, we immediately get, cf.~\eqref{eq:direct_sum_N},
\begin{equation}
  \label{eqn-first-sum}
  \cA(G) = \tcA_{N} \oplus \tcA_{ZN} \oplus \tcA_{ZZ} \oplus d\RR^V.
\end{equation}

\usetikzlibrary{decorations.pathreplacing}
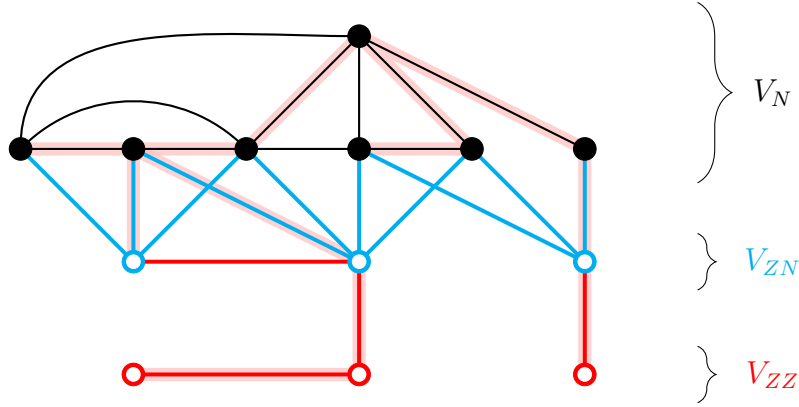
\begin{figure}[H]
\begin{tikzpicture}[scale = 1.5]

\foreach \x in {1,...,6}
  {\node [inner sep = 0pt](N\x) at (\x,4){};}
\node[inner sep = 0pt] (N0) at (4,5){};

\foreach \x in {1,2,3}
  {\node[inner sep = 0pt](ZN\x) at (2*\x,3){};}

\foreach \x in {1,...,3}
{\node[inner sep = 0pt](ZZ\x) at (2*\x,2){};}

\draw[line width =5, color =pink, opacity = .7](N1)--(N3)--(N0);
\draw[line width =5, color =pink, opacity = .7](ZZ1)--(ZZ2) --(ZN2)--(N2)--(ZN1);
\draw[line width =5, color =pink, opacity = .7](N4)--(N5)--(N0)--(N6)--(ZN3)--(ZZ3);

\draw[thick] (N1) to [out = 45, in = 135] (N3);
\draw[thick] (N1) to [out = 90, in = 180] (N0);
\draw[thick] (N1) -- (N5);
\foreach \x in {3,...,6}
{\draw[thick] (N0) -- (N\x);}

\foreach \x in {2,...,5}
  {\draw[line width=1.5, color=cyan] (ZN2) -- (N\x);}
\foreach \x in {1,...,3}
  {\draw[line width=1.5, color = cyan] (ZN1) -- (N\x);}
\foreach \x in {4,...,6}
  {\draw[line width=1.5, color = cyan](ZN3) -- (N\x);}

\draw[line width=1.5, color=red] (ZN1) -- (ZN2)--(ZZ2)--(ZZ1);
\draw[line width=1.5, color=red] (ZN3) -- (ZZ3);

\foreach \x in {0,...,6}
{\draw[fill] (N\x) circle [radius=.1];}

\foreach \x in {1,...,3}
{\draw[fill, color=cyan] (ZN\x) circle [radius = .1];
  \draw[fill, color=white] (ZN\x) circle [radius =.06];}

\foreach \x in {1,...,3}
{\draw[fill, color=red] (ZZ\x) circle [radius = .1];
\draw[fill, color=white] (2*\x,2) circle [radius = .06];}

  \draw[decorate,decoration={brace,mirror,  amplitude = 13}]    (7,3.7)--(7,5.3) node [midway, xshift = 1cm] {$V_N$}    ;
  \draw[decorate,decoration={brace,mirror,  amplitude = 7}]    (7,2.75)--(7,3.25) node [midway, xshift = 1cm, cyan] {$V_{ZN}$}    ;
  \draw[decorate,decoration={brace,mirror,  amplitude = 7}]    (7,1.75)--(7,2.25) node [midway, xshift = 1cm, red] {$V_{ZZ}$}    ;

\end{tikzpicture}
\caption{A spanning tree (shaded edges) consistent with
  \S\ref{subsec:tree_choice}}\label{fig-tree}
\end{figure}

As explained in \S\ref{subsec-lift}, the space $\tcA_{N} \oplus
\tcA_{ZN} \oplus \tcA_{ZZ}$ is the universal covering space of
$\tcM_h$ and provides a convenient set of coordinates for it.

For future reference we note the dimensions of each subspace,
\begin{align}
  \label{eq:Atilde_dim_N}
  &\dim \tcA_N
  = \left| E_N^\f\right| = |E_{NN}| - |V_N| + 1,
  \\
  \label{eq:Atilde_dim_ZN}
  &\dim \tcA_{ZN}
  = \left| E_{ZN}^\f\right| = |E_{ZN}| - |V_{ZN}|,
  \\
  \label{eq:Atilde_dim_ZZ}
  &\dim \tcA_{ZZ}
  = \left| E_{ZZ}^\f\right| = |E_{ZZ}| - |V_{ZZ}|.
\end{align}

We now give conditions under which the space $\tcA_{N}$ can be
replaced by another space in the decomposition \eqref{eqn-first-sum},
to be used in Lemma \ref{lem:kerB} below.

\begin{prop}
  \label{prop:W_replacement}
  Let $W_N \subset\cA_{N}$ be a subspace such that
  \begin{equation}
    \label{eq:AN_decomposition}
    \cA_N = W_N \oplus d_N \RR^{V_N},
  \end{equation}
  where $d_{N}:\RR^{V_{N}}\to\cA_{N}$ is the coboundary
  operator of $G_{N}$.  Then
  \begin{equation}\label{eqn-direct-sum}
    \cA(G) = W_N \oplus  \tcA_{ZN} \oplus \tcA_{ZZ} \oplus d_G\RR^V.
  \end{equation}
\end{prop}

\begin{proof}
  From \eqref{eq:AN_decomposition} we conclude that $\dim W_N$ is the
  Betti number of the subgraph $G_N$. On the other hand, the restriction
  of our choice (\S\ref{subsec:tree_choice}) of the spanning tree to $G_N$
  is a spanning tree of $G_N$, which means that
  \begin{equation*}
    \dim(\tcA_N) = \dim(W_N).
  \end{equation*}
  The dimensions of the four subspaces on the right-hand side of
  \eqref{eqn-direct-sum} add up to $\dim \cA(G)$. The subspaces
  $W_N$, $\tcA_{ZN}$, and $\tcA_{ZZ}$ are supported on different parts
  of the graph, so they are mutually orthogonal. We are left with showing
  that the following intersection is trivial:
  \begin{equation}
    \label{eq:d_intersection}
    d_G\RR^V \cap \Bigl(W_N \oplus \tcA_{ZN} \oplus \tcA_{ZZ}\Bigr) = \{0\}.
  \end{equation}

  Suppose $\alpha \in \cA(G)$ is in the intersection \eqref{eq:d_intersection}, i.e.
  \begin{equation*}
    \alpha = \alpha_N + \alpha_Z,
    \quad \alpha_N\in W_N,
    \quad \alpha_Z\in \tcA_{ZN}\oplus \tcA_{ZZ},
  \end{equation*}
  and also $\alpha\in d_G\RR^V$. The latter implies that the flux (see \S \ref{subsec-cohomology}) of
  $\alpha$ through any cycle $\xi$ is $\int_\xi \alpha = 0$.  For any
  cycle $\xi$ in $G_N$, we have $\int_\xi \alpha_Z = 0$ since $\xi$
  and $\alpha_Z$ have disjoint support. Hence,
  $\int_\xi \alpha_N = 0$.  This means that
  $\alpha_N \in d_N\RR^{V_N} \cap W_N$, so $\alpha_N = 0$ by the
  assumption \eqref{eq:AN_decomposition}. We conclude that
  \begin{equation*}
    \alpha = \alpha_Z \in \Bigl(\tcA_{ZN} \oplus \tcA_{ZZ}\Bigr) \cap d_G\RR^V,
  \end{equation*}
  which is zero by \eqref{eqn-first-sum}.
\end{proof}


\section{Proof of \texorpdfstring{Theorem \ref{thm:main} parts
  \eqref{item:Ftopology}--\eqref{item:Fcritical}}{the main theorem, parts (1)--(2)}}
\label{subsec-proof-topology}
\subsection{Constructing the critical set \texorpdfstring{$F$}{}} \label{sec:constructingF}
Recall the notion of the critical data $(V_N, h_N, \psi_N, \lambda)$ of
\S\ref{subsec-data} and the notation $\psi$ for the extension of
$\psi_N$ by 0 to the rest of the graph.  Since $\cM_h = \cM_{h'}$ for
any signing $h'$ of $h$, without loss of generality we may assume that
$h_N$ is the restriction $h|G_N$ (rather than the restriction of a
signing of $h$ or, equivalently, signing of a restriction).

Using the notation introduced in \S\ref{subsec-magnetic-perturbation},
the definition \eqref{eq:F} of the critical set $F$ can be rewritten
as
\begin{equation}
  \label{eq:F_alpha}
  F = \{[\alpha*h]\in \cM_h \colon \alpha\in\cA(G),\
  (\alpha*h) |G_N = h_N, \  (\alpha*h) \psi = \lambda \psi \}.
\end{equation}
As a first step, we observe that it is enough to search for $\alpha$
in the special gauge choice $\tcA(G)$ based on the tree constructed
from the given $V_N$.

\begin{lem}
  \label{lem:critF}
  The critical set $F=F(V_N, h_N, \psi_N, \lambda)$ can be computed as
  \begin{equation}
    \label{eq:F_alpha_tree}
    F
    = \{[\ta*h] \colon \ta\in\tcA,\
    (\ta*h) |G_N = h_N,\  (\ta*h) \psi = \lambda \psi \},
  \end{equation}
  where $\tcA = \tcA_{N} \oplus \tcA_{ZN} \oplus \tcA_{ZZ}$ is
  constructed as in \S\ref{subsec:tree_choice}, based on the set $V_N$.
\end{lem}

\begin{proof}
  The inclusion right to left is immediate from \eqref{eq:F_alpha} and
  $\tcA \subset \cA(G)$.
  We now need to show that if $\alpha \in \cA(G)$ is such that
  $\alpha*h$ satisfies the conditions in \eqref{eq:F_alpha}, there
  exists $\ta \in \tcA$ so that $\ta*h$ satisfies the same conditions.

  By \eqref{eqn-first-sum} we have
  \begin{equation}
    \label{eq:alpha_decomp}
    \alpha = \alpha_N + \alpha_Z - d\theta,
    \quad
    \alpha_N \in \tcA_N,\ \alpha_Z \in \tcA_{ZN} \oplus \tcA_{ZZ}.
  \end{equation}
  Due to the condition $(\alpha*h) |G_N = h_N = h|G_N$, the matrix
  $\alpha$ is zero (modulo $2\pi$) on all edges in $E_{NN}$.  For
  any cycle $\xi$ in $G_N$, $\int_\xi d\theta = 0$ and
  \begin{equation}
    \label{eq:cycle_int}
    0 = \int_\xi \alpha
    =  \int_\xi \left( \alpha_N + \alpha_Z - d\theta \right)
    = \int_\xi \alpha_N, 
  \end{equation}
  since $\alpha_Z$ and $\xi$ have disjoint support.  We conclude that
  $\alpha_N = 0$.  Consequently, $d\theta$ is 0 on every edge in
  $E_{NN}$ and $\theta$ is constant on the vertices in $V_N$.
  Furthermore, due to the kernel of $d$,  $\theta$ can be chosen to be
  0 on $V_N$.

  Now let $\ta = \alpha + d\theta$.  Denoting
  $h_\alpha = \alpha*h$ and recalling that
  $(\alpha+d\theta)*h = d\theta * h_\alpha$, we have
  \begin{equation*}
    (\ta*h)|G_N = (d\theta * h_\alpha) | G_N = h_\alpha | G_N,
  \end{equation*}
  because $d\theta$ is 0 on the edges of $G_N$.  Furthermore,
  $e^{\pm i\theta} \psi  = \psi$ since $\theta$ and $\psi$ have
  disjoint supports and therefore, cf.~\eqref{eq:dtheta_action},
  \begin{equation*}
    (\ta*h)\psi = (d\theta * h_\alpha) \psi = e^{i\theta} h_\alpha e^{-i\theta} \psi
    = e^{i\theta} h_\alpha \psi = \lambda e^{i\theta} \psi = \lambda \psi.\qedhere
  \end{equation*}
\end{proof}

\subsection{Proof of \texorpdfstring{\Cref{thm:main} part
    \eqref{item:Ftopology}}{main theorem part (1)}}
  According to \Cref{lem:critF} we are looking for $\ta =\ta_N +
  \ta_{ZN} + \ta_{ZZ} \in \tcA_{N}
  \oplus \tcA_{ZN} \oplus \tcA_{ZZ}$ satisfying the two conditions
  \begin{align}
    \label{eq:Fcond1}
    &(\ta*h) |G_N = h|G_N,\\
    \label{eq:Fcond2}
    &(\ta*h) \psi = \lambda \psi.
  \end{align}
  We have already seen that condition~\eqref{eq:Fcond1} is equivalent
  to $\ta_N=0$.  Since $\ta_{ZZ}$ is supported on edges $(rs)\in E$
  such that $\psi_r = \psi_s = 0$, condition \eqref{eq:Fcond2} imposes
  no constraints on $\ta_{ZZ}$.  To see which constraints are imposed
  on $\ta_{ZN}$, we inspect the rows of $(\ta*h) \psi = \lambda \psi$
  corresponding to $r\in V_{ZN}$, namely
  \begin{equation}
    \label{eqn-inhomogeneous}
    h_{rr'}\psi_{r'}
    + \sum_{s: (rs) \in E_{ZN}^\f} e^{i\alpha_{rs}} h_{rs}
    \psi_s
    = 0 = \lambda \psi_r,
  \end{equation}
  where $(rr')$ is the unique edge incident to $r$ in $E_{ZN} \setminus E_{ZN}^\f$.
  The solutions $\{\alpha_{rs}\}$ to Equation~\eqref{eqn-inhomogeneous} constitute the 
  linkage space $M(b_r)$ (see Appendix \ref{subsec-linkage})  with
  \begin{equation}
    \label{eq:linkage_br}
    b_r = \left( \left|h_{rr'}\psi_{r'}\right|,
      \ \left|h_{rs_1}\psi_{s_1}\right|,
      \ \left|h_{rs_2}\psi_{s_2}\right|, \ldots \right).
  \end{equation}
  The linkage space is a manifold because the genericity requirement
  \eqref{eqn-generic-linkage} is satisfied by
  Lemma~\ref{lem:generic_linkage}.  For different choices of
  $r \in V_{ZN}$ the corresponding equations involve distinct variables,
  therefore the manifolds $M(b_r)$ are disjoint.

  We conclude that the set of solutions $\ta$ of
  \eqref{eq:Fcond1}-\eqref{eq:Fcond2} modulo $2\pi$ is diffeomorphic to the product
  \begin{equation}
    \label{eq:F_description}
    \{0\} \times \prod_{r \in V_{ZN}}M(b_r) \, \times \TT_{ZZ} \cong F,
  \end{equation}
  where $\TT_{ZZ}$ is the torus of antisymmetric matrices supported on
  $E_{ZZ}^\f$ with entries modulo $2\pi$; its dimension is given by
  \eqref{eq:Atilde_dim_ZZ}.

  According to \Cref{thm-linkage}, the linkage space $M(b_r)$ is
  non-empty if 
  \begin{equation}
    \label{eq:linkage_nonzero}
    2 \max_{s\in V_{N}} |h_{rs}\psi_{s}|
    \le 
    \sum_{s\in V_{N}}|h_{rs}\psi_{s}|.
  \end{equation}
  If all $M(b_r)$ are non-empty, the dimension of $F$ is 
  \begin{align*}
    \dim(F)
    &= |E_{ZZ}|-|V_{ZZ}|
      +\sum_{r\in V_{ZN}}\Big(\big|\{(rs)\in E : s\in V_N\}\big| - 3\Big)\\
    &=|E_{ZZ}|-|V_{ZZ}|+|E_{ZN}|-3|V_{ZN}|,
  \end{align*}
  in agreement with \eqref{eq:MBdimension}, which finishes the proof of \Cref{thm:main}
  part \eqref{item:Ftopology}.

  \subsection{Proof of \texorpdfstring{\Cref{thm:main} part
      \eqref{item:Fcritical}}{main theorem, part (2)}}

  Each point $h_\alpha \in F$ satisfies the criticality conditions
  \eqref{eqn-criticality} because those need only be checked on edges
  $rs \in E_{NN}$ where $h_\alpha$ coincides with $h_N$, which is real
  with a real eigenvector $\psi_N$.  Conversely, starting with a
  critical point $h_\alpha$ with a simple eigenvector $\psi$, denote
  by $V_N$ the support of $\psi$.  Following \cite[Thm 3.2(3) and
  \S5.2]{AloGor_jst23}, let $\theta_s$ be such that
  $e^{i \theta_s} \psi_s$ is real for all $s$.  From the criticality
  condition, equation \eqref{eqn-criticality}, we conclude that
  \begin{equation}
    \label{eq:sub_criticality}
    (h_\alpha)_{rs} e^{i(\theta_r - \theta_s)} \in \RR
    \qquad\text{for all }r,s \in V_N.
  \end{equation}
  In other words, $\theta$ defines a gauge transformation
  $e^{i\theta} h_\alpha e^{-i\theta}$ whose restriction to $V_N$ is real
  and whose eigenvector $e^{i \theta} \psi$ is also real.  Without loss
  of generality we can now take $h_\alpha$ to be this representative.

  We now establish that $V_N$ is an admissible support.  We observe
  that $h_{\alpha}|G_N$ is a signing of $h|G_N$ and that
  $\psi_{N}:=\psi|V_N$ is its eigenvector with eigenvalue $\lambda$.
  Assumption \S \ref{ass} \ref{item-subset} implies that this is a
  simple eigenvalue.  Consequently, the induced graph $G_N$ is
  connected: otherwise $\psi_{N}$ (which is nowhere zero on $G_N$ by
  construction) would restrict to an eigenvector on each connected
  component of $G_N$, contradicting the simplicity of $\lambda$.

  To check the last condition of \S\ref{subsec-admissible} we observe
  that if $r\not\in V_N$ has a neighbor in $V_N$, the eigenvalue
  equation satisfied by $\psi$ at $r$ is a planar linkage equation of
  the form
  \eqref{eqn-linkage}.  The linkage is generic by \Cref{lem:generic_linkage}
  and, by \Cref{thm-linkage}, must have at least 3 terms to
  possess a solution.  This finishes the proof of \Cref{thm:main}
  part \eqref{item:Fcritical}.

\section{Spectral shift}

To compute the Morse index in part \eqref{item:MBindex} of
Theorem~\ref{thm:main} we need to establish some auxiliary results.

\subsection{Spectral shift}\label{subsec-spectral-shift}
Let $h$ be a Hermitian matrix on the space $\CC^V$ and let $\psi$ be
an eigenvector of $h$ corresponding to a simple eigenvalue $\lambda =
\lambda_k(h)$.  Assume $V_0 \subset V$ is such that $\psi(v_0) = 0$
for all $v_0 \in V_0$.  Let $V_1$ be the complement of $V_0$.  Note
that we do not require that $\psi$ is everywhere non-vanishing on $V_1$.

Expressing $h \psi = \lambda \psi$ with respect to the
decomposition $\CC^V = \CC^{V_0} \oplus \CC^{V_1}$, we write
\begin{equation}
  \label{eqn-h}
  h\psi =
  \begin{pmatrix}
    A & B \\
    B^* & D
  \end{pmatrix}
  \begin{pmatrix}
    0 \\ \psi_1
  \end{pmatrix}
  = \lambda
  \begin{pmatrix}
    0 \\ \psi_1
  \end{pmatrix}.
\end{equation}
This implies that $D\psi_1 = \lambda \psi_1$ but, in general,
$\lambda= \lambda_k(h)$ will have a different position in the spectrum
of $D$, say $D\psi_1 = \lambda_{k'}(D) \psi_1$.  The difference $k-k'$
is called the {\em spectral shift}; by the Cauchy interlacing theorem
it is an integer between $0$ and $|V_0|$.  Spectral shift appears
in the formula \eqref{eq:MBindex} for the Morse index and in this
section we seek an analytical handle on it.

If the assumption of eigenvalue simplicity is dropped (which is
\emph{not} the case needed in \Cref{thm:main}), the spectral shift is
defined more robustly using the number of eigenvalues \emph{below}
$\lambda$.  If $M$ is a Hermitian matrix we write
$n_+(M), n_-(M), n_0(M)$ for the number of positive, negative and zero
eigenvalues respectively.  Then the {\em spectral shift} is defined as
\begin{equation}
  \label{eq:spec_shift_def}
  n_-(h-\lambda I) - n_-(D-\lambda I).
\end{equation}

\subsection{Pseudoinverse} \label{subsec-pseudoinverse} Recall that
the Moore--Penrose pseudoinverse $M^+$ of a matrix $M$ is
uniquely determined by the four conditions that $M^+MM^+ = M^+$,
$MM^+M = M$, and $MM^+$, $M^+M$ are Hermitian.
If $T$ is an invertible
matrix then checking the four conditions proves that
\begin{equation}
  \label{eqn-TMT}
  \left(TMT^*\right)^+ = \left(T^{-1}\right)^*M^+T^{-1}.
\end{equation}

If $M$ is itself Hermitian, which is the setting in which we will use $M^+$,
the pseudoinverse is the inverse of $M$ on the orthogonal complement of $\ker M$,
extended to the rest of the space by zero.  In particular,
\begin{equation}
  \label{eq:MP_index}
  n_*\left(M^+\right) = n_*(M),
  \qquad * \in \{+,-,0\}.
\end{equation}
Furthermore, in the
Hermitian case, $M^+ M = M
M^+$ is the orthogonal projector onto $\Ran M$.

\begin{prop}
  \label{prop-compression}
  Let $h$ be a Hermitian matrix with a block decomposition as in
  \eqref{eqn-h} above, and an eigenvector $\psi = (0,\psi_1)$ with
  eigenvalue $\lambda$.  Assume $\ker(D-\lambda I) \subset \ker B$
  (which holds, in particular, if $\psi_1$ is a simple eigenvector of
  $D$), and denote by $R_0(\lambda)$ the compression\footnote{The
    compression of a linear operator $T:\CC^m \to \CC^m$ to a subspace
    $V$ is the composition
    \[ V \overset{\iota}{\longrightarrow} \CC^m
      \overset{T}{\longrightarrow} \CC^m \overset{P}{\longrightarrow}
      V\] where $\iota$ is the inclusion and $P=\iota^*$ is the
    orthogonal projection.  It is the (1,1) component of the matrix of
    $T$ with respect to the decomposition $V \oplus V^{\perp}= \CC^n$
  } of $(h - \lambda I)^+$ to $V_0$.  Then
  \begin{equation}
    \label{eq:spec_shift}
    n_*\big(R_0(\lambda)\big)
    = n_*(h-\lambda I) - n_*(D-\lambda I),
    \qquad * \in \{+,-,0\}.
  \end{equation}
\end{prop}

\begin{proof}
  Without loss of generality we may assume $\lambda = 0$.  The result will follow from two facts:
  \begin{enumerate}
  \item\label{item:Haynsworth} An extension of the Haynsworth identity
    \cite{Hay_laa68} to non-invertible $D$, cf.\
    \cite[App.~A]{BerCanCoxMar_paa22},
      \begin{equation}\label{eq:Haynsworth}
        n_*(h) = n_*(D) + n_*(A-BD^+B^*),
        \qquad * \in \{+,-,0\}.
      \end{equation}
    \item\label{item:Compression} The compression of $h^+$ to $V_0$ is $(A-BD^+B^*)^+$.
  \end{enumerate}
  These two facts give the desired equality:
  \begin{equation}
    \label{eq:Haynsworth_used}
    n_*\left( (A-BD^+B^*)^+ \right)
    = n_*\left( A-BD^+B^* \right)
    = n_*(h) - n_*(D).
  \end{equation}
  
  To establish \eqref{eq:Haynsworth}, the usual proof of the Haynsworth
  identity, define
  \[
    T =
    \begin{pmatrix}
      I & BD^+ \\
      0 & I
    \end{pmatrix}
    \qquad 
    M =
    \begin{pmatrix}
      A - BD^+B^* & 0 \\
      0 & D
    \end{pmatrix},
  \]
  and calculate
  \[
    TMT^* =
    \begin{pmatrix}
      A - BD^+B^* + BD^+DD^+B^* & BD^+D \\
      DD^+B^* & D
    \end{pmatrix}
    = h.
  \]
  Here we used $D^+DD^+ = D^+$ (\S \ref{subsec-pseudoinverse}) and
  $BD^+D = B(I-P) = B$, where $P$ is the projection onto
  $\ker D \subset \ker B$, and therefore $BP = 0$.  Since $T$ is
  invertible, \eqref{eq:Haynsworth} follows from the Sylvester Law of
  Inertia.

  For item~\eqref{item:Compression} the compression of $h^+$ is its (1,1) component. Use \eqref{eqn-TMT} to find
  \begin{equation*}
    h^+ = 
    \begin{pmatrix}
      I & 0 \\
      - D^+B^* & I
    \end{pmatrix}
    \begin{pmatrix}
      (A - BD^+B^*)^+ & 0 \\
      0 & D^+
    \end{pmatrix}
    \begin{pmatrix}
      I & -BD^+ \\
      0 & I
    \end{pmatrix}
    =
    \begin{pmatrix}
      (A - BD^+B^*)^+ & * \\
      * & *
    \end{pmatrix}.
  \end{equation*}
\end{proof}

\section{The Hessian along a critical manifold \texorpdfstring{$F$}{}}
\label{sec:theHessian}

This section proves Theorem \ref{thm:main},
part~\ref{item:MBindex} by computing the Hessian of the eigenvalue
$\lambda_{k}: \cM_h \to \RR$ at a critical point $[h_\alpha]$. 
Recall the setting of \Cref{thm:main}:
\begin{itemize}
\item The matrix $h\in\cS(G)$ is generic in the sense of \S\ref{ass}.
\item The $k$-th eigenvalue $\lambda=\lambda_{k}(h_\alpha)$ of the
  matrix $h_\alpha\in[h_\alpha]\in\cM_h$ is simple, with the
  normalized eigenvector $\psi$.
\item The eigenvector $\psi$ is supported on $V_{N}$, which has been
  shown \S\ref{subsec-proof-topology} to be an admissible support
  \S\ref{subsec-admissible}.  The corresponding partition $V =
  V_N \sqcup V_{ZN} \sqcup V_{ZZ}$ gives rise to induced
  subgraphs $G_N$, $G_{ZN}$ and $G_{ZZ}$.
\item Denote $h_N = h_\alpha | G_N$ and $\psi_N= \psi|V_N$; the
  vector $\psi_N$ is an eigenvector of $h_N$ with eigenvalue
  $\lambda$.
\item The criticality condition \eqref{eqn-criticality} holds:
  $(h_\alpha)_{rs}\psi_{r}\bar{\psi}_{s}\in\RR$ for all edges
  $(rs)\in E$.  This implies, see \eqref{eq:sub_criticality}, that the
  representative $h_\alpha$ may be chosen so that $\psi_N$ and $h_N$
  are real.
\item By assumption, $\lambda$ is a simple eigenvalue of $h_N$.  
  Assume also that $\lambda$ is not in the spectrum of the
  restriction $h_{\alpha}|G_{ZZ}$. Denote the number of
  eigenvalues of $h_{\alpha}|G_{ZZ}$ that are smaller than $\lambda$
  by $k_{ZZ}$
\end{itemize}

Although we are interested in the Hessian of $\lambda_{k}$ as a function on the quotient space $\tcM_h = \cP_h / \sim$, it turns out to be beneficial (as seen in \cite{AloGor_jst23,Col_apde13}) to compute the Hessian on the total space $\cP_h$.

The latter Hessian is a bilinear form on $\cA(G)$, and the
computation is based on finding a ``nice'' decomposition of $\cA(G)$
in which the Hessian is block-diagonal.  The appropriate decomposition
turns out to be based on the spanning tree chosen in
\S\ref{subsec:tree_choice}, and is of the form \eqref{eqn-direct-sum}
with the subspace $W_N$ defined in the next subsection.

\subsection{Properties of the first derivative \texorpdfstring{$B$}{}}
The derivative of $\lambda_{k}$ at the point $h_\alpha\in \cP_h$
is given by the Hellman--Feynman formula (see \Cref{sec-derivatives}),
\begin{equation}
  \label{eq:HF_formula}
  \partial_\gamma \lambda_k
  = \left<\psi, \partial_\gamma h_\alpha \psi\right>
  = \left<\psi, B \gamma\right>,
  \qquad \gamma \in T_{h_\alpha} \cP_h \simeq \cA(G),
\end{equation}
where
\begin{equation}
  \label{eq:B_op_def}
  B: \cA(G) \to \CC^V,
  \qquad
  B\gamma:= \partial_\gamma \left(h_\alpha \psi\right)
  = \frac{d}{dt} \Big[ h_{\alpha+t\gamma} \psi \Big]_{t=0}.
\end{equation}
Note that $B$ is an operator
from a real space to a complex one.  Explicitly,
\begin{equation}
  \label{eqn-B}
  (B\gamma)_r
  = \frac{d}{dt} \sum_{\{s: r \sim s \} } e^{it\gamma_{rs}}
  \left(h_\alpha\right)_{rs} \psi_s
  = i \sum_{\{s: r \sim s \} } \gamma_{rs} \left(h_\alpha\right)_{rs}\psi_s.
\end{equation}
Now define the subspace $W_{N}\subset\cA_{N}$ to be the kernel of the restriction $B|\cA_{N}$,
\begin{equation}
  \label{eq:WN_def}
  W_N :=\{\gamma\in\cA_{N}\ :\ B\gamma=0\}.  
\end{equation}

\begin{lem}\label{lem:kerB} The following decomposition holds:
\begin{equation}
  \label{eq:WN_decomposition_again}
  \cA(G)
  = W_N \oplus  \tcA_{ZN} \oplus \tcA_{ZZ} \oplus d\RR^V. 
\end{equation}
\end{lem}

\begin{proof}
  The claim follows from \Cref{prop:W_replacement} and the fact that
  \[\cA_{N}  = W_{N}\oplus d_{N}\RR^{V_{N}},\]
  where $d_{N}:\RR^{V_{N}}\to\cA_{N}$ is the coboundary operator of
  the induced subgraph $G_N$.  The latter fact has been established in
  \cite[Lem.~5]{Col_apde13} or \cite[\S 5.4(3) and \S
  5.7]{AloGor_jst23}\footnote{In these references, $W_N$ is denoted
    $\ker d^*$ and $V_H$ respectively. To give an idea of the proof,
    note that in the special case $(h_\alpha)_{rs} = -1$ and
    $\psi_s=1$, the operator $iB|\cA_N$ is the adjoint of $d_N$ with
    respect to the standard inner products; the general case follows
    by customizing an inner product to retain the latter property.},
  which are applicable because $B|\cA_{N}$ is defined through
  \eqref{eqn-B} only in terms of the real $h_{N}$ and its
  non-vanishing eigenvector $\psi|V_{N}$.
\end{proof}

\begin{prop}
  \label{prop:B_properties}
  In the decomposition of the vectors spaces  
  \[\cA(G)=W_{N} \oplus \tcA_{ZN} \oplus \tcA_{ZZ} \oplus d\RR^V,
    \quad\text{and}\quad
    \CC^{V}=\CC^{V_{N}}\oplus\CC^{V_{ZN}}\oplus\CC^{V_{ZZ}}, \]
  the operator  $B: \cA(G) \to \CC^V$ has the following block structure:
  \begin{equation}
    \label{eq:B-blocks}
    B = \quad
    \begin{blockarray}{ccccc}
      &{\scriptstyle W_{N}} & {\scriptstyle \tcA_{ZN}}
      & {\scriptstyle \tcA_{ZZ}} & {\scriptstyle d\RR^V} \\
      \begin{block}{c[cccc]}
        {\scriptstyle \CC^{V_{N}}}  & 0 & 0 & 0 & * \\ 
        {\scriptstyle \CC^{V_{ZN}}} & 0 & * & 0 & * \\
        {\scriptstyle \CC^{V_{ZZ}}} & 0 & 0 & 0 & 0 \\
      \end{block}
    \end{blockarray}.
 \end{equation} 
Furthermore,  
  \begin{enumerate}  
  \item \label{item-Bdtheta} $B$ maps $d\RR^V$ to
    $\CC^{V_N}\oplus\CC^{V_{ZN}}$ as follows: for
    $\theta \in \RR^V$,
    \begin{equation}
      \label{eq:Bdtheta}
      Bd\theta = i(h_\alpha-\lambda)\Theta\psi,
      \qquad
      \Theta = \diag(\theta_1, \theta_2, \cdots, \theta_n).
    \end{equation}

    \item \label{item:Bsurjective}
      $B$ maps $\tcA_{ZN}$ surjectively onto $\CC^{V_{ZN}}$.
    \end{enumerate}
\end{prop}

We remark that part (\ref{item:Bsurjective}) is deeper than it may
seem: it establishes that a linear operator maps a real space
surjectively onto a complex space.

\begin{proof} First, we prove the block structures of \eqref{eq:B-blocks}:
\begin{itemize}
\item The first column is zero by the definition of $W_{N}$.
\item If $\gamma\in\tcA_{ZZ}$ then $B\gamma = 0$ (i.e.\ the third
  column is zero): indeed, $\gamma_{rs}$ is non-zero only if $(rs)\in
  E_{ZZ}$, and then $\psi_r = \psi_s=0$; therefore all terms in \eqref{eqn-B}
  are zero.
\item If $r \in V_{ZZ}$ then $(B\gamma)_{r}=0$ (i.e.\ the third
  row is zero): indeed, $\psi_s=0$ for any $s\sim r$ by the definition of $V_{ZZ}$.
\item If $\gamma\in\tcA_{ZN} $ and $r\in V_{N}$ then $(B\gamma)_r=0$:
  indeed, $\psi_s$ is non-zero only if $s\in V_{N}$, and then $\gamma_{rs}=0$.
\end{itemize}

To prove (\ref{item-Bdtheta}), differentiate
$td\theta * h_\alpha = e^{-it\Theta}h_\alpha e^{it\Theta}$ to get
\begin{equation}
  \label{eqn-Bdf}
  Bd\theta = i(h_\alpha \Theta - \Theta h_\alpha)\psi
  = i(h_\alpha-\lambda)\Theta \psi.
\end{equation}

To prove (\ref{item:Bsurjective}),  notice that the different
components of the vector $B\gamma\in\CC^{V_{ZN}}$, for
$\gamma\in\tcA_{ZN}$, depends on disjoint sets of $\gamma$
coordinates. Therefore, it is enough to show that for any
$r\in V_{ZN}$ the map $\gamma\mapsto (B\gamma)_{r}$ is surjective onto
$\CC$.  This map from $\tcA_{ZN}$ to $\CC$ is, by definition, the
derivative of $f(\alpha'):=(h_{\alpha'}\psi)_{r}$ at the point
$\alpha$. The equation $f(\alpha')=0$ is a linkage equation, which is
moreover generic by \Cref{lem:generic_linkage}.  By
\Cref{thm-linkage}(\ref{item:surjective_differential}), the derivative
of $f$ is surjective onto $\CC$ at every zero point, including the
point $\alpha'=\alpha$, where $f(\alpha)=\lambda \psi_{r}=0$.
\end{proof}

\subsection{The Hessian}
The Hessian of $\lambda_k(h_\alpha)$ is a bilinear form on $\cA(G)$,
which can be expressed as\footnote{For the reader's convenience,
  the calculation is recalled in \Cref{sec-derivatives} and equation \eqref{eqn-Hessian1}.}
\begin{equation}
  \label{eqn-Hessian}
  \Hess \lambda_k(h_\alpha)   =  \Qout + \Qin,
\end{equation}
where the bilinear forms $\Qout, \Qin: \cA(G) \to \RR$ are defined as follows:
\begin{align}
  \label{eq:Qout_def}
  \Qout(\gamma,\delta)
  &= -2\Re \left< B\delta, (h_\alpha-\lambda)^+B\gamma
    \right>_{\CC^V}, \\
  \nonumber
  \Qin(\gamma,\delta)
  &= \left< \psi, \partial^2_{\gamma,\delta}h_\alpha\psi
    \right>_{\CC^V}\\
  \label{eq:Qin_def}  
  &=- \sum_{(rs)\in E} \gamma_{rs}\delta_{rs}(h_\alpha)_{rs}\psi_r\bar \psi_s.
\end{align}

\begin{example}
  \label{ex:Hessian_calc}
  Before analyzing the Hessian in the general case, we illustrate the
  formula \eqref{eqn-Hessian}--\eqref{eq:Qin_def} in the particular
  case of the constant eigenvector of the standard
  Laplacian:
  \begin{equation}
    \label{eq:standardLaplacian}
    h_{rs} =
    \begin{cases}
      -1, &r\sim s,\\
      \deg(r), & r=s,\\
      0, &\text{otherwise},
    \end{cases}
    \qquad \qquad
    \lambda=\lambda_1 = 0,
    \qquad
    \psi = \frac1{\sqrt{|V|}}
    \begin{pmatrix}
      1\\
      \vdots\\
      1
    \end{pmatrix}.
  \end{equation}
  Under these simplifying assumptions, the subspace $W_N$ consists of
  1-forms $\gamma$ that satisfy,
  \begin{equation}
    \label{eq:WN_example}
    \sum_{\{s: r \sim s \} } \gamma_{rs} = 0
    \qquad \text{for all } r\in V,
  \end{equation}
  cf.\ \eqref{eqn-B}-\eqref{eq:WN_def}.  The operator on the left hand
  side of \eqref{eq:WN_example} is the dual\footnote{Here we use the
    standard inner products on $\RR^V$ and $\cA(G)$.} $d^*$ of the
  coboundary operator $d: \RR^V \to \cA(G)$, so that $W_N = \ker d^*$.
  The space $W_N$ is the space of flows on the graph, or, using the
  duality between $C_1(G,\RR)$ and $C^1(G, \RR)$, see
  \S\ref{subsec-cohomology}, the space of cycles $H_1(G,\RR)$.

  Using assumptions \eqref{eq:standardLaplacian} in equations
  \eqref{eqn-Hessian}-\eqref{eq:Qin_def}, the Hessian on
  $\cA(G) = \ker d^* \oplus d \RR^V$ is the bilinear form
  \begin{equation}
    \label{eq:Hess_ex_answer}
    \big(\Hess \lambda_1\big) (\gamma+ d\theta, \delta + d \varphi)
    = \big(\Hess \lambda_1\big) (\gamma, \delta)
    = \sum_{(rs)\in E} \gamma_{rs} \delta_{rs} = \gamma \cdot \delta.
  \end{equation}
  Choosing a basis of the space $W_N \simeq H_1(G,\RR)$, the
  Hessian in this basis is the \emph{Gram matrix} of this basis.  If the basis is
  integer-valued, the corresponding matrix is sometimes called the
  \emph{intersection matrix} \cite{KotSun_aam00}; its determinant is
  equal to the number of spanning trees of the graph, see
  \cite[Prop.~26.3]{Big_blms97} or \cite[Sec.~3]{KotSun_aam00}.

  Dropping the condition that $h$ is the standard Laplacian (but
  keeping the requirement that the constant $\psi = (1,\ldots,1)^t$ is
  an eigenvector), the above computation may be repeated to produce the
  \emph{cycle intersection matrix} of Bronski, DeVille and Ferguson
  \cite{BroDevFer_siamjam16}, who used a statement equivalent to the nodal
  magnetic theorem (\Cref{rem:signings}) to obtain the spectral
  position $k$ of a known eigenvector and thus analyze stability
  of a network of nonlinear coupled oscillators.
\end{example}

\begin{prop}
  \label{prop-QplusQ}
  With respect to the decomposition in \eqref{eq:WN_decomposition_again},
  \[\cA(G)
  = W_N \oplus  \tcA_{ZN} \oplus \tcA_{ZZ} \oplus d_G\RR^V,\]  
  the bilinear forms $\Qin$ and $\Qout$ 
  have the following block structure:
  \begin{equation}
    \label{eq:Hessian_terms}
    \Qin =
    \begin{pmatrix}
      \Qin|{W_N} & 0 & 0 & 0 \\
      0 & 0 & 0 & 0\\
      0 & 0 & 0 & 0\\
      0 & 0 & 0 & -\Qout|d\RR^{V}
    \end{pmatrix},
    \qquad
    \Qout =
    \begin{pmatrix}
      0 & 0 & 0 & 0 \\
      0 & \Qout|\tcA_{ZN} & 0 & 0\\
      0 & 0 & 0 & 0\\
      0 & 0 & 0 & \Qout|d\RR^{V}
    \end{pmatrix}.
  \end{equation}
\end{prop}

\begin{proof}
  If $\gamma \in \tcA_{ZN} \oplus \tcA_{ZZ}$ then
  $\gamma_{rs}\psi_{s}\bar\psi_{r}=0$ for every edge $(rs)$. This
  shows the second and third rows and columns of $\Qin$ vanish.  If
  $\gamma \in W_N \oplus \tcA_{ZZ}$ then $B\gamma = 0$ by
  \Cref{prop:B_properties}, therefore first and third rows and columns
  of $\Qout$ vanish. The fact that $\lambda_k$ is constant under
  perturbations by $d\theta$ means that the fourth row and column of
  the total Hessian $\Qin+\Qout$ are zero.  To summarize,
  \begin{equation}\label{eq: dR Hess}
    \Qin+\Qout = 
    \begin{pmatrix} * & 0 & 0 & * \\ 0 & 0 & 0 & 0\\
      0 & 0 & 0 & 0\\
      * & 0 & 0 & *\end{pmatrix} +
    \begin{pmatrix} 0 & 0 & 0 & 0 \\ 0 & * & 0 & *\\
      0 & 0 & 0 & 0\\
      0 & * & 0 & *\end{pmatrix}=
    \begin{pmatrix} * & * & * & 0 \\ * & * & * & 0\\
      * & * & * & 0\\
      0 & 0 & 0 & 0\end{pmatrix}.
  \end{equation}
 It follows that the off-diagonal entries marked * in
  $\Qin$ and $\Qout$ also vanish and that their
  $(4,4)$-entries add up to 0.
\end{proof}

\begin{lem}
  \label{lem:Qout_dR_index} The restriction of $\Qout$ to $d\RR^V$ has  \begin{align}
    \label{eq:Qout_on_dR_index}
    &\ind \left( -\Qout | d\RR^V \right) = k_N-1,\\
    \label{eq:Qout_on_dR_kernel}
    & \rank \left( -\Qout | d\RR^V \right) = |V_N|-1.
  \end{align}
\end{lem}

\begin{proof}
  If $\theta_1, \theta_2 \in \RR^V$ then \eqref{eq:Bdtheta}  and the
  properties of the pseudoinverse, \S\ref{subsec-pseudoinverse}, give
  \begin{align*}
    -\Qout(d\theta_1,d\theta_2)
    &= 2 \Re \left< Bd\theta_1, (h_\alpha-\lambda)^+Bd\theta_2\right>\\
    &= 2\Re\left< (h_\alpha-\lambda)\Theta_1 \psi,
      (h_\alpha-\lambda)^+(h_\alpha-\lambda)\Theta_2\psi \right> \\
    & = 2 \Re\left< \Theta_1\psi, (h_\alpha-\lambda)\Theta_2\psi \right>.
  \end{align*}
  Since $\psi$ is non-vanishing and real on $V_N$ and $\Theta$ are
  arbitrary real diagonal matrices, the vectors $\Theta_j \psi$ span
  the real space $\RR^{V_N}$.  The compression of $h_\alpha-\lambda$
  to $\RR^{V_N}$ is the matrix $(h_\alpha-\lambda)|G_N=h_{N}-\lambda$,
  which is real, so $-\Qout | d\RR^V $ has the same rank and index as
  $h_{N}-\lambda$.  Since $\lambda$ is the simple eigenvalue number
  $k_{N}$ of $h_{N}$, the matrix $h_N-\lambda$ has rank $|V_N|-1$ and $k_N-1$
  negative eigenvalues.
\end{proof}

\begin{lem}
  \label{lem:index-Qout-ZN}The restriction of $\Qout$ to $\tcA_{ZN}$ has 
  \begin{align}
    \label{eq:Qout-ZN-ind}
    &\ind\big(\Qout|\tcA_{ZN}\big) = 2|V_{ZN}| - 2(k-k_N-k_{ZZ}),\\
    \label{eq:Qout-ZN-ker}
    &\dim \ker \big(\Qout|\tcA_{ZN}\big) = |E_{ZN}| - 3|V_{ZN}|.
  \end{align}
\end{lem}

\begin{proof}
  Write
  $\Qout(\gamma,\delta) = \Re\left< B\gamma,
    -2(h_\alpha-\lambda)^+B\delta\right>$ and recall that by
  \Cref{prop:B_properties}(\ref{item:Bsurjective}), the operator
  $B$ maps $ \tcA_{ZN}$ surjectively onto $\CC^{V_{ZN}}$, which means that the kernel of $B$ has (real) dimension $\dim(\tcA_{ZN})-2|V_{ZN}|=|E_{ZN}|-3|V_{ZN}|$.
  Applying \Cref{prop-Qreal} with $H$ being the
  compression of $-2(h_\alpha-\lambda)^+$ to
  $\CC^{V_{ZN}}$, gives
\begin{align*}
    &\ind\big(\Qout|\tcA_{ZN}\big) 
    = 2\ind_\CC \big(H\big),\\
    &\dim \ker \big(\Qout|\tcA_{ZN}\big) 
    = 2\dim_\CC \ker \big(H\big)+|E_{ZN}|-3|V_{ZN}|.
  \end{align*}
We now aim to use \Cref{prop-compression} with $V_0 = \CC^{V_{ZN}}$ and $V_1
= \CC^{V_N}\oplus \CC^{V_{ZZ}}$.  The compression $D$ of $h=h_\alpha$
to $V_1$ is block diagonal because $V_{N}$ and $V_{ZZ}$ are not
connected by any edges.
This means that $\lambda$ is a simple eigenvalue number
$k'=k_{N}+k_{ZZ}$: $n_0(D-\lambda)=1$ and $n_-(D-\lambda) =
k_N+k_{ZZ}-1$.  From $n_0(h_\alpha-\lambda) = 1$ and
$n_-(h_\alpha-\lambda) = k-1$ we conclude that the compression of
$(h_\alpha-\lambda)^+$ to $\CC^{V_{ZZ}}$ has no kernel and has
$k-k_{N}-k_{ZZ}$ negative eigenvalues, so
\begin{align*}
  &\dim_\CC \ker \big(H\big)=0,\\
  &\ind_\CC \big(H\big)=|V_{ZN}|-(k-k_{N}-k_{ZZ}).\qedhere
\end{align*}
\end{proof}

\begin{lem}
  \label{lem:index_Qin_WN}  
 The restriction of $\Qin$ to $W_{N}$ has
  \begin{align}
    \label{eq:Qin_ind}
    &\ind\big(\Qin|W_N\big) = \nodcnt(h_\alpha,k) - (k_N-1) = \sigma(h_N,k_N), \\
    &\dim\ker\big(\Qin|W_N\big) = 0.
  \end{align}
\end{lem}

\begin{proof}  
  The mixed partial derivatives $\partial^2_{\gamma,\delta}h$ vanish
  whenever $\gamma, \delta$ have disjoint support.  So the matrix for
  $\Qin$ is diagonal with respect to the coordinates given
  by the edges of the graph.  Explicitly, if $\gamma = \delta$ is the basis element corresponding to the edge $(rs) \in E$, 
  \begin{equation}
    \label{eq:Qin_diagonal}
    \Qin(\gamma,\delta)
    = -(h_{rs}\psi_r \bar \psi_s + h_{sr} \psi_s \bar \psi_r)
    = -2 \Re(h_{rs}\psi_r \bar \psi_s)
    = -2h_{rs}\psi_r\bar\psi_s,
  \end{equation}
  which is real by the criticality condition~\eqref{eqn-criticality}.
  The index of the diagonal matrix $\Qin$ is
  the number of negative diagonal entries, namely number of edges for which $h_{rs}\psi_r\bar\psi_s>0$. Since $\psi$ is supported on $V_{N}$ we need to count only $E_{N}$ edges, which means this is exactly the nodal count
  $\nodcnt(h_N,k_{N})$ as $\psi_{N}$ is the $k_{N}$-th eigenvector of $h_{N}$.  From
  \eqref{eq:Hessian_terms} and \eqref{eq:Qout_on_dR_index} we get
  \begin{equation}
    \label{eq:Qin_WN_compute}
    \ind\left(\Qin|W_N\right) = \ind(\Qin)
    - \ind\left( -\Qout | d\RR^V\right) =
    \nodcnt(h_N,k_{N}) - (k_N-1)=\sigma(h_N,k_{N}).
  \end{equation}
  The rank of $\Qin$ is the number of non-zero entries of the form
  \eqref{eq:Qin_diagonal}, namely $|E_{NN}|$.  Using
  \eqref{eq:Hessian_terms} and \Cref{lem:Qout_dR_index} again, we get
  \begin{equation*}
    \rank \left(\Qin|W_N\right) = \rank(\Qin)
    - \rank\left( -\Qout | d\RR^V\right) = 
    |E_{NN}| - \big(|V_N| - 1\big),   
  \end{equation*}
  which matches the dimension of $W_N$.  Therefore $\Qin|W_N$ is full rank.
\end{proof}

\subsection{Proof of Theorem \ref{thm:main}, part~\ref{item:MBindex}}
\label{subsec-index-Hessian}

We will compute the index of the Hessian, establishing
\eqref{eq:MBindex}.  Then we will compute the nullity of the Hessian,
to show it is non-degenerate in the directions transversal to $F$.

By \Cref{prop-QplusQ}, the Hessian of $\lambda_k$ as a function on
$\cP_h$ is 
\begin{equation}
  \label{eq:Hessian_blocks}
  \Hess_{\cP_h}\left( \lambda_k(h_\alpha) \right)
  =
  \begin{pmatrix}
    \Qin|{W_N} & 0 & 0 & 0 \\
    0 & \Qout|\tcA_{ZN} & 0 & 0\\
    0 & 0 & 0 & 0\\
    0 & 0 & 0 & 0
  \end{pmatrix}.
\end{equation}
To compute the Hessian of $\lambda_k$ on $\tcM_h$, we recall from
\S\ref{subsec-lift} that the tangent space of $\tcM_h$ is $\tcA(G)$.
Comparing \eqref{eq:direct_sum_N} with
\eqref{eq:WN_decomposition_again} we conclude that we need to ignore
the last row and column of \eqref{eq:Hessian_blocks}, i.e.
\begin{equation}
  \label{eq:Hessian_Mh}
  \Hess_{\tcM_h}\left( \lambda_k(h_\alpha) \right)
  =
  \begin{pmatrix}
    \Qin|{W_N} & 0 & 0  \\
    0 & \Qout|\tcA_{ZN} & 0\\
    0 & 0 & 0
  \end{pmatrix}.
\end{equation}

Therefore, using \Cref{lem:index-Qout-ZN} and \Cref{lem:index_Qin_WN}, 
\begin{align*}
  \ind \Hess_{\tcM_h}\left( \lambda_k \right)
  &= \ind\left(\Qin|W_N\right) + \ind\left(\Qout|\tcA_{ZN}\right)\\
  &=\sigma(h_{N},k_{N})+2|V_{ZN}|-2(k-k_{N}-k_{ZZ}),  
\end{align*}
proving \eqref{eq:MBindex}.

Similarly, we compute
\begin{align*}
  \dim \ker \Hess_{\tcM_h}(\lambda_k)
  &= \dim\ker\big(\Qin|_{W_N}\big) + \dim\ker\big(\Qout|\tcA_{ZN}\big)
    + \dim \tcA_{ZZ} \\
  &= 0 + |E_{ZN}| - 3|V_{ZN}| + |E_{ZZ}| - |V_{ZZ}|
    = \dim(F),
\end{align*}
where we used \Cref{lem:index_Qin_WN}, \ref{lem:index-Qout-ZN},
\eqref{eq:Atilde_dim_ZZ} and \eqref{eq:MBdimension} for the four
dimensions appearing in the computation.  In conclusion, the Hessian
is non-degenerate on any direct complement of the tangent space of
$F\subset \tcM_h$, as claimed.

\section{Proof of Theorem \ref{thm-existence}: existence of critical points}\label{proof-existence}

\subsection{}  We are given
\begin{enumerate}  
\item $G = (V,E)$ with admissible support $V_N$ and decomposition $V = V_N \cup V_{ZN}\cup V_{ZZ}$, 
\item $h_N \in \cS(G_N)$ with simple eigenvalue
  $\lambda = \lambda_{k_N}(h_N)$ and real eigenvector $\psi$,
\item $h_{ZN} \in \cS(G_{ZN})$ arbitrary, satisfying $\lambda \notin
  \spec(h_{ZN})$; $k_{ZN} = n_-(h_{ZN} - \lambda)$,
\item $h_{ZZ} \in \cS(G_{ZZ})$ arbitrary, satisfying $\lambda \notin
  \spec(h_{ZZ})$; $k_{ZZ} = n_-(h_{ZZ} - \lambda)$.
\end{enumerate}
The aim is to find $h \in \cS(g)$ which agrees with $h_N$, $h_{ZN}$
and $h_{ZZ}$ on the corresponding subgraphs, and $h_\alpha \in \tcM_h$
which is a critical point of $\lambda_{k_N+k_{ZN}+k_{ZZ}}$.

\subsection{}
For each $r \in V_{ZN}$ choose a collection $H_{rs} \in \RR$ (where $(rs) \in 
E_{ZN}$) such that
\begin{align}
  \label{eq:linkage_solvable}
  &|H_{rt}\psi_t| < \sum_{(rs) \in E_{ZN},\ s \ne t} |H_{rs}\psi_s|
    \quad \text{for all }(rt)\in E_{ZN}, \\
  \label{eq:linkage_generic}
  &\sum_{(rs) \in E_{ZN}} \epsilon_{s}|H_{rs}\psi_s| \neq 0
    \quad \text{for all }\epsilon_s=\pm1.
\end{align}
Note that a valid choice of $H_{rs}$ is always available (make all
$|H_{rs} \psi_s| \approx 1$ and rationally independent) and that any
scalar multiple of a valid collection $\{H_{rs}\}$ will also satisfy
\eqref{eq:linkage_solvable}.  Setting $H_{rs}=0$
whenever $s\in V_N$ but $(r,s) \notin E$, we obtain a
$|V_{ZN}| \times |V_N|$ matrix $H$.

By \Cref{thm-linkage}, inequalities
\eqref{eq:linkage_solvable}-\eqref{eq:linkage_generic} guarantee
existence of a non-zero (modulo $\pi$) solution $\{\alpha_{rs}\}$ to
the equation
\begin{equation}
  \label{eq:linkage_solved}
  \sum_{s:\ (rs)\in E_{ZN}} e^{i \alpha_{rs}}H_{rs}\psi_s =  0,  
\end{equation}
attached to the vertex $r$.  By adjusting the overall phase, we can
make one of the $\alpha_{rs}=0$; we do it for the unique edge
$(rr')\in E_{ZN} \setminus E^\f_{ZN}$.  Similarly set
$\alpha_{rs} = 0$ whenever $(rs) \notin E_{ZN}$ to obtain a non-zero
$\alpha \in \tcA_{ZN}$.

Finally, we introduce an arbitrary $|V_{ZN}| \times |V_{ZZ}|$ real
matrix $B$ such that $B_{sr} \neq 0$ if and only if
$(sr) \in E$.

\subsection{}
We claim that, for sufficiently small $\epsilon$, the existence
problem is solved by the following matrix, written here with respect
to the decomposition $V = V_N \sqcup V_{ZN} \sqcup V_{ZZ}$,
\[  h_\epsilon =
\begin{pmatrix}
h_N & \epsilon.H^t & 0\\
\epsilon.H & h_{ZN} & \epsilon.B^t\\
0 & \epsilon.B & h_{ZZ}
\end{pmatrix} 
\quad\text{and}\quad
h_{\epsilon,\alpha} = 
\begin{pmatrix} 
h_N & \epsilon.H_\alpha^* & 0\\
\epsilon.H_\alpha & h_{ZN} & \epsilon.B^t\\
0 & \epsilon.B & h_{ZZ}
\end{pmatrix}.
\]
When $\epsilon=0$ the matrix $h_{0,\alpha}$ is block diagonal and has a
simple eigenvalue $\lambda$ with the label $k_N+k_{ZN}+k_{ZZ}$; its
eigenvector is
\[ 
  \begin{pmatrix}
    \psi_N \\ 0 \\ 0
  \end{pmatrix}.
\]
By \eqref{eq:linkage_solved}, the same eigenvalue and eigenvector pair
remain valid for $h_{\epsilon,\alpha}$ for all $\epsilon$; by
continuity, $\lambda$ remains simple and retains its label
$k_N+k_{ZN}+k_{ZZ}$ for small $\epsilon$.  The criticality condition
\eqref{eqn-criticality} clearly holds since $\psi_N$ is real.

Finally, the genericity conditions of \S \ref{ass} are open and dense,
while conditions
\eqref{eq:linkage_solvable}-\eqref{eq:linkage_generic} are open.
Therefore, perturbing the non-zero entries of $h_\epsilon$ will result
in a generic $h_\epsilon$, with the rest of the construction still
being valid.
This concludes the proof of Theorem \ref{thm-existence}.  \qed

\section{Examples}
\label{sec:examples}

In this section we explore examples of graphs which illustrate that
conditions of part~\eqref{item:MBindex} of \Cref{thm:main} cannot be
weakened in general.  The conclusion is that the critical manifold $F$
is not, in general, Morse--Bott \cite{Bot_am54}: the negative
eigenspace of the Hessian may fail to form a vector bundle over $F$.

\subsection{Change of the index on \texorpdfstring{$F$}{the critical
    manifold} while the eigenvalue is simple}
\label{sec:example1}

\begin{figure}
  \centering
  \includegraphics[scale=1.0]{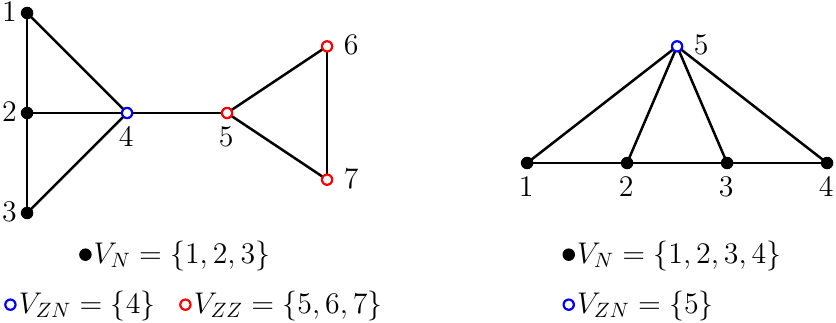}
  \caption{(Left) The graph in the example of \S\ref{sec:example1}.
    (Right) The graph in the example of \S~\ref{sec:ex2}.}
  \label{fig:ex_graphs}
\end{figure}

Consider a matrix $h$ strictly supported on the graph $G$ in
\Cref{fig:ex_graphs}(left), such that up to gauge equivalence,
$\tcM_h$ is parametrized by
\begin{equation}
  \label{eq:ha_ex1}
  h_\alpha :=
  \begin{pmatrix}
    1 &-1 &             0 &           -1 &          0 &   0 &           0\\
    -1 & 2 &          -1 &           -e^{i\alpha_1} & 0 &   0 &           0\\
    0 & -1 &           1 &           -e^{i\alpha_2} & 0 &   0  &          0\\
    -1 &-e^{-i\alpha_1} &-e^{-i\alpha_2} &  4 &         -1 &   0 &           0\\
    0 &  0 &           0 &           -1 &          1 &  -1 &          -1\\
    0 &  0 &           0 &            0 &         -1 &   2 &          -e^{i\alpha_3}\\
    0 &  0 &           0 &            0 &         -1 &  -e^{-i\alpha_3} & 2
  \end{pmatrix}
\end{equation}

Choosing $V_N = \{1,2,3\}$, we get $V_{ZN} = \{4\}$ and $V_{ZZ} =
\{5,6,7\}$.  We consider
\begin{equation}
  \label{eq:hN_ex1_def}
  h_N :=
  \begin{pmatrix}
    1 &-1 &             0 \\
    -1 & 2 &          -1 \\
    0 & -1 &           1 
  \end{pmatrix},
  \qquad
  \psi_N =
  \begin{pmatrix}
    1 \\ 1 \\ 1
  \end{pmatrix},
  \qquad
  \lambda = 0 = \lambda_1(h_N),
\end{equation}
so that the critical manifold is
\begin{align}
  \label{eq:critF_ex1}
  F &= F(V_N, h_N, \psi_N, \lambda)
  \cong M\big( [1, 1, 1]\big) \times \TT \\
  \nonumber
  &= \left\{ \left(\frac{2\pi}{3}, -\frac{2\pi}3, \alpha_3\right)
    \colon \alpha_3 \in \RR/ 2\pi\ZZ \right\}
  \sqcup
  \left\{ \left(-\frac{2\pi}{3}, \frac{2\pi}3, \alpha_3\right)
    \colon \alpha_3 \in \RR/ 2\pi\ZZ \right\},
\end{align}
in the coordinates $(\alpha_1, \alpha_2, \alpha_3)$ implicit in
\eqref{eq:ha_ex1}.

Numerical inspection of the spectrum of $h_\alpha$ on $F$, see
\Cref{fig:ex1_alpha3_spectra}(left), confirms that
$0 = \lambda_2(h_\alpha)$ is a simple eigenvalue, fulfilling one of
the two conditions of \Cref{thm:main}(\ref{item:MBindex}).

\begin{figure}
  \centering
  \includegraphics[scale=0.45]{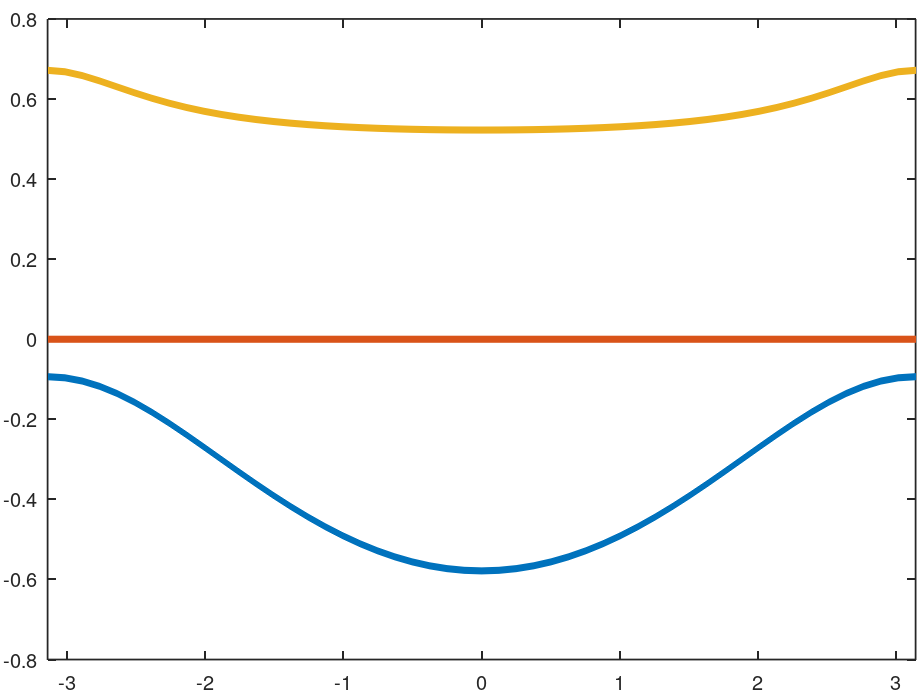}
  \includegraphics[scale=0.45]{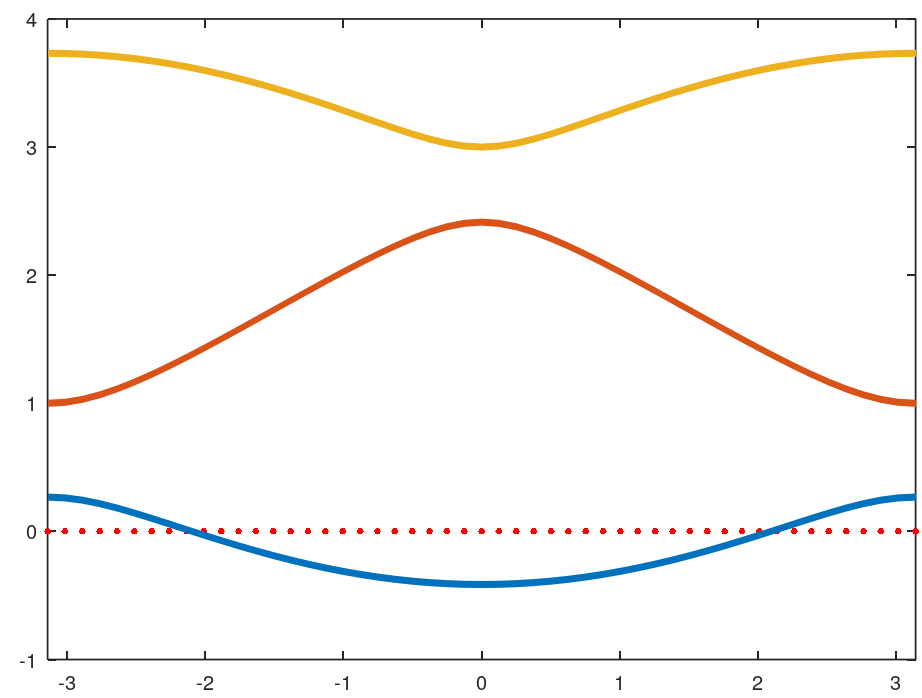}
  \caption{(Left) First three eigenvalues of $h_\alpha$ restricted to
    $F$, as functions of $\alpha_3$, confirming that
    $0=\lambda_2(h_\alpha)$ is simple on $F$.  (Right) Spectrum of
    $h_\alpha|G_{ZZ}$ is a function of $\alpha_3$; the dotted line is
    added at $\lambda=0$ to guide the eye, confirming that
    $0\in\spec\left(h_\alpha|G_{ZZ}\right)$ at two values of $\alpha_3$.}
  \label{fig:ex1_alpha3_spectra}
\end{figure}

The restriction of $h_\alpha$ to $G_{ZZ}$, the graph induced by
$V_{ZZ}$, still contains one phase,
\begin{equation}
  \label{eq:ZZ_restriction_ex1}
  h_\alpha|G_{ZZ} =
  \begin{pmatrix}
    1 &  -1 &          -1\\
    -1 &   2 &          -e^{i\alpha_3}\\
    -1 &  -e^{-i\alpha_3} & 2
  \end{pmatrix},
\end{equation}
therefore the condition
$0=\lambda \not\in \spec\left(h_\alpha |G_{ZZ}\right)$ appearing in
\Cref{thm:main}(\ref{item:MBindex}) may fail for some
$h_\alpha \in F$.  This is what happens in this example, as
illustrated in Figure~\ref{fig:ex1_alpha3_spectra}(right).  We denote
by $\pm\alpha_{3,c}$ the two points where
$0 \in \spec\left(h_\alpha |G_{ZZ}\right)$.  The data entering the
index formula \eqref{eq:MBindex} is $\sigma(h_N, k_N) = 0$, 
$|V_{ZN}| = 1$, $k = 2$, $k_N = 1$ and 
\begin{equation}
  \label{eq:ex1_index_data}
  k_{ZZ} =
  \begin{cases}
    1 & \alpha_3 \in (-\alpha_{3,c}, \alpha_{3,c}), \\
    0 & \alpha_3 \in (-\pi, -\alpha_{3,c}) \cup (\alpha_{3,c}, \pi).
  \end{cases}
\end{equation}
The resulting index in the direction normal to $F$ is
\begin{equation}
  \label{eq:ex1_ind}
  \ind(\lambda_2, h_\alpha) = 2 k_{ZZ} =
  \begin{cases}
    2 & \alpha_3 \in (-\alpha_{3,c}, \alpha_{3,c}), \\
    0 & \alpha_3 \in (-\pi, -\alpha_{3,c}) \cup (\alpha_{3,c}, \pi).
  \end{cases}
\end{equation}
This answer is confirmed by inspecting the eigenvalue
$\lambda_2(h_\alpha)$ as a function of $(\alpha_1, \alpha_2)$ at
$\alpha_3=0$ and $\alpha_3=\pi$, see \Cref{fig:ex1_eig_surf}.

\begin{figure}
  \centering
  \includegraphics[scale=0.45]{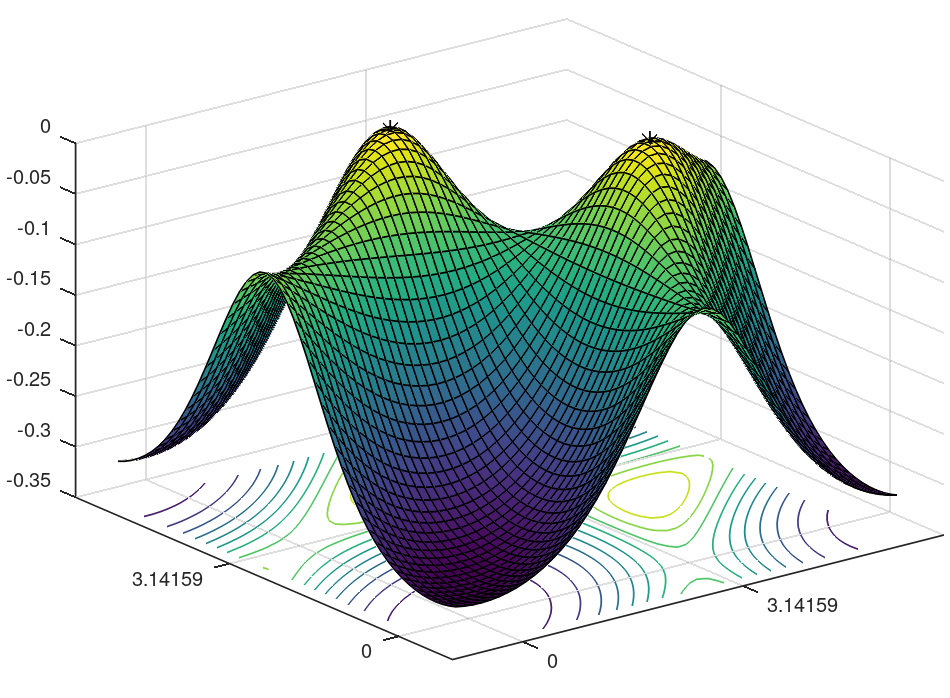}
  \includegraphics[scale=0.45]{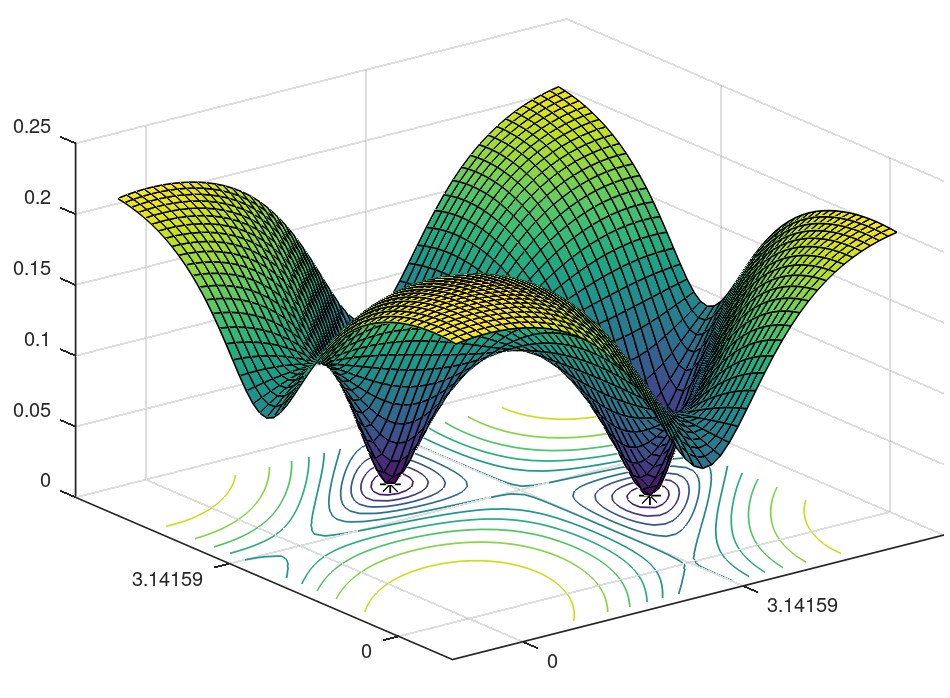}
  \caption{Eigenvalue $\lambda_2(h_\alpha)$ on slices normal to $F$:
    namely, $\alpha_1$ and $\alpha_2$ are varied while $\alpha_3$ is
    kept fixed at $0$ (left) or $\pi$ (right).  In both cases we have
    critical points $(\alpha_1,\alpha_2)=\left(\pm\frac{2\pi}3,
      \mp\frac{2\pi}3\right)$ with critical value $0$ (highlighted by
    star symbols on both graphs).  These points are maxima when
    $\alpha_3=0$ and minima when $\alpha_3=\pi$, in agreement with
    \eqref{eq:ex1_ind}.}
  \label{fig:ex1_eig_surf}
\end{figure}

Finally, when $\alpha_3 = \pm\alpha_{3,c}$, the eigenvalue
$\lambda_2(h_\alpha)$ as a function of $(\alpha_1, \alpha_2)$ is
identically 0.

\subsection{Change of the index on \texorpdfstring{$F$}{the critical
    manifold} due to multiplicity}
\label{sec:ex2}

The following example shows that $\lambda_k(h_\alpha)$ can
become multiple on $F$, breaking the first of
the two conditions of \Cref{thm:main}(\ref{item:MBindex}).  The graph
on which $h$ is strictly supported is shown in \Cref{fig:ex_graphs}(right).
In order to illustrate the logic behind our construction, the matrices
$h_\alpha$ are specified up to a free parameter $\gamma \in \RR$,
\begin{equation}
  \label{eq:ex2_matrix}
  h_\alpha =
  \begin{pmatrix}
    1 & -1 & 0 & 0 & -e^{i\alpha_1} \\
    -1 & 3 & -2 & 0 & -e^{i\alpha_2} \\
    0 & -2 & 10 & -4 & -e^{i\alpha_3} \\
    0 & 0 & -4 & 2 & -1\\
    -e^{-i\alpha_1} & -e^{-i\alpha_2} & -e^{-i\alpha_3} & -1 & \gamma
  \end{pmatrix}.
\end{equation}
Choose the admissible support $V_N=\{1,2,3,4\}$ and therefore
$V_{ZN} = \{5\}$ and $V_{ZZ} = \emptyset$.  Consider
\begin{equation}
  \label{eq:ex2_hN}
  h_N :=
  \begin{pmatrix}
    1 & -1 & 0 & 0 \\
    -1 & 3 & -2 & 0 \\
    0 & -2 & 10 & -4 \\
    0 & 0 & -4 & 2
  \end{pmatrix},
  \qquad
  \psi_N =
  \begin{pmatrix}
    1 \\ 1 \\ 1 \\ 2
  \end{pmatrix},
  \qquad
  \lambda = 0 = \lambda_1(h_N),
\end{equation}
so that the critical manifold is
\begin{align}
  \label{eq:critF_ex2}
  F = F(V_N, h_N, \psi_N, \lambda)
  &\cong M\big( [1, 1, 1, 2]\big)
  \\ \nonumber
  &=
  \left\{ (\alpha_1,\alpha_2,\alpha_3) \in (\RR / 2\pi \ZZ)^3
  \colon e^{i\alpha_1} + e^{i\alpha_2} + e^{i\alpha_3} +2 = 0\right\}. 
\end{align}
We stress that the critical manifold is present and identical for all
values of the parameter $\gamma$.  The linkage manifold $M\big( [1, 1,
1, 2]\big)$ is topologically a circle and is illustrated in \Cref{fig:linkage1112}.

\begin{figure}
  \centering
  \includegraphics[scale=0.75]{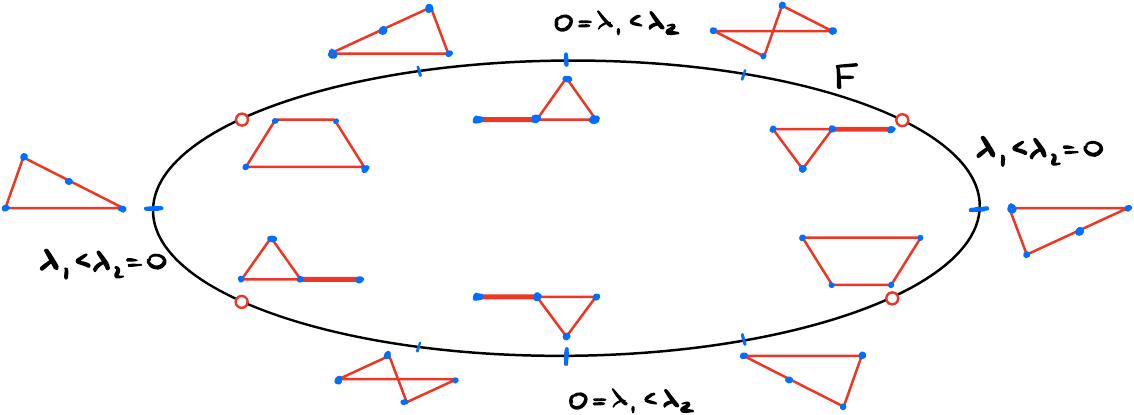}
  \caption{A schematic depiction of the linkage space from
    \eqref{eq:critF_ex2}.  The point \eqref{eq:ex2_alpha_ex} is the
    top left empty circle on the manifold $F$; the other points where
    $\lambda=0$ is a double eigenvalue are also denoted by empty
    circles.}
  \label{fig:linkage1112}
\end{figure}

Fixing a point in $F$, for example
\begin{equation}
  \label{eq:ex2_alpha_ex}
  \alpha = \left(\frac{2\pi}3, \pi, -\frac{2\pi}3\right)  
\end{equation}
in parametrization \eqref{eq:ex2_matrix}, we consider the spectrum of
$h_\alpha$ as a function of $\gamma$.  For different values of
$\gamma \in \RR$, these matrices are rank-1 perturbations of each
other (and of $h_N$), and it can be shown that the ordered eigenvalues
$\lambda_k(\gamma) = \lambda_k\big(h_\alpha(\gamma)\big)$ have the
following properties (see \Cref{fig:ex2_spec}(left)):
\begin{itemize}
\item Each $\lambda_1(\gamma),\ldots, \lambda_1(\gamma)$ is
  continuous and nondecreasing in $\gamma$.
\item For all $k=1,\ldots,4$,
  \begin{align*}
    &\lim_{\gamma \to -\infty}
    \lambda_{k+1}(\gamma)
    = \lim_{\gamma \to \infty}
    \lambda_k(\gamma) = \lambda_k(h_N) \\
    &\lim_{\gamma \to -\infty} \lambda_1(\gamma) =
    -\infty,
    \qquad
    \lim_{\gamma \to \infty} \lambda_5(\gamma) =
    \infty.
  \end{align*}
\end{itemize}
Since $\lambda_1(h_N)=0$ is in the spectrum of $h_\alpha(\gamma)$ for all
$\gamma$, we conclude that it is the eigenvalue $\lambda_2(\gamma)$
for $\gamma\to-\infty$ and the eigenvalue $\lambda_1(\gamma)$ for
$\gamma\to\infty$.  Therefore, there must be a value of $\gamma$ for
which $\lambda_1(\gamma) = \lambda_2(\gamma) = 0$.  If we use this
$\gamma$ in the definition of $h$, then we are guaranteed at least one
point on the critical manifold $F$ with a multiple eigenvalue $0$.

\begin{figure}
  \centering
  \includegraphics[scale=0.45]{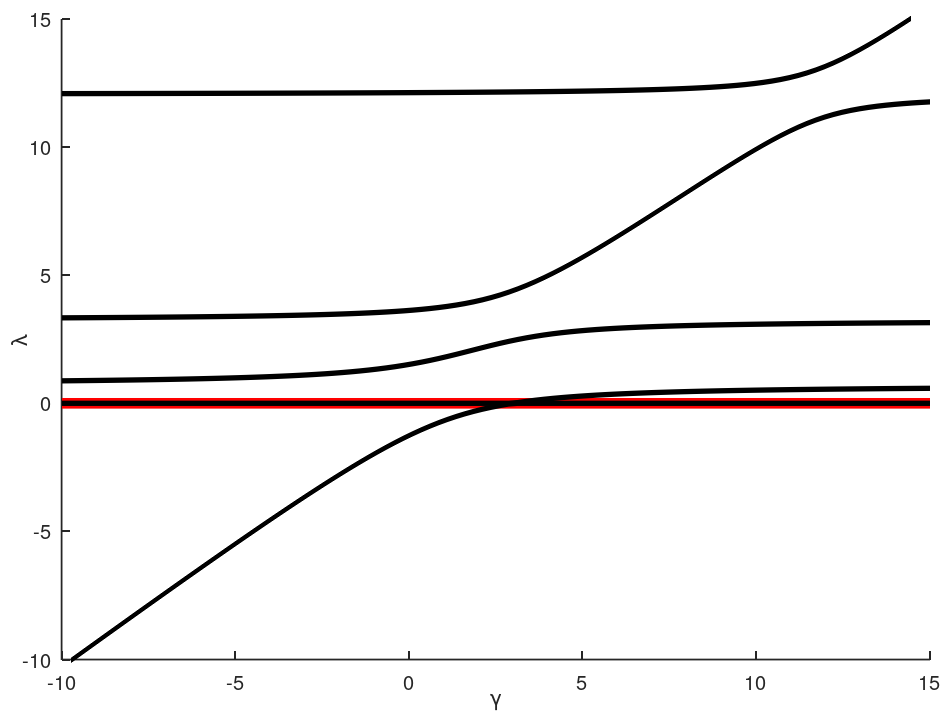}
  \includegraphics[scale=0.45]{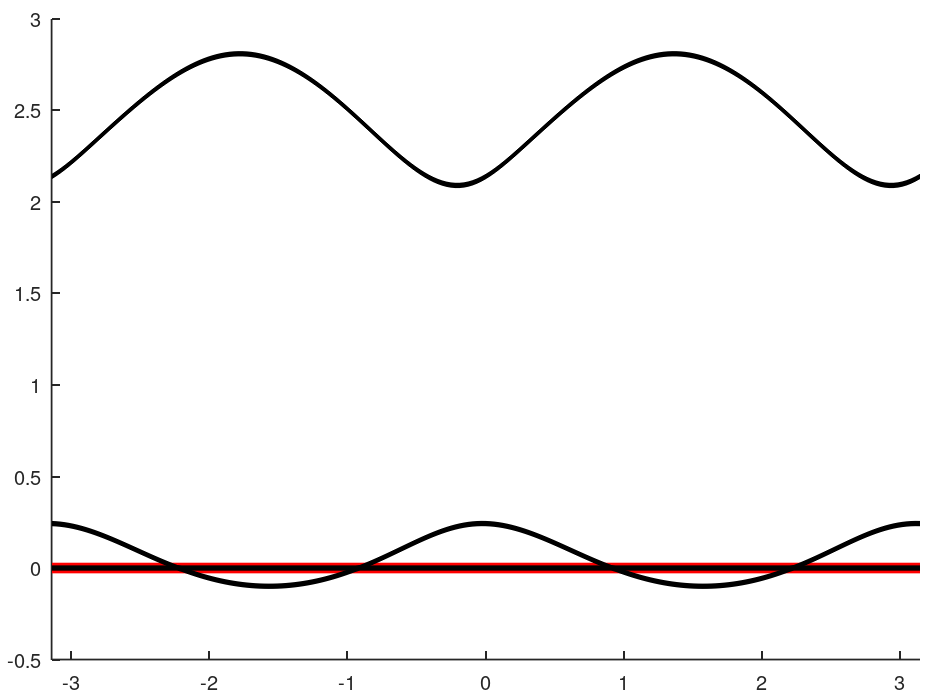}
  \caption{(Left) Spectrum of $h_\alpha$ with $\alpha$ given by
    \eqref{eq:ex2_alpha_ex}, as a function of the parameter $\gamma$.
    The value $0$ (highlighted) is in the spectrum for all
    $\gamma$. (Right) The first three eigenvalues of $h_\alpha$ along
    the critical submanifold $F$.  The eigenvalue $0$ (highlighted)
    alternates between being $\lambda_1$ and $\lambda_2$.}
  \label{fig:ex2_spec}
\end{figure}

For the particular choice \eqref{eq:ex2_alpha_ex}, this value of
$\gamma$ turns out to be $\gamma=3$, see \Cref{fig:ex2_spec}(left).
Fixing $\gamma$, we can inspect the spectrum of $h_\alpha$ along $F$,
see \Cref{fig:ex2_spec}(right).  Since the critical value $\lambda=0$ alternates between
being eigenvalue number 1 and eigenvalue number 2, we conclude that
the presence of points of multiplicity on the critical manifold $F$
persist under small perturbations of the underlying matrix $h$.

\section{Universality of nodal distribution and critical points}
\label{sec-experiments}

In this section the universal nodal distribution conjecture is stated.
Its history is reviewed.  Experimental evidence for this conjecture is provided
in examples of $3$-regular graphs with many non-symmetry critical points. 

\subsection{Numerical results and conjectures}\label{intro-distribution}

The motivation for understanding the nodal count distribution came
from the Bogomolny--Schmit conjecture, due to Smilansky, Gnutzmann,
and Blum \cite{BluGnuSmi_prl02} and Bogomolony, and Schmit
\cite{bogomolny2002percolation}, that predicts a Gaussian behavior of
the appropriately normalized nodal count, as we range over all
eigenfunctions of a quantum chaotic system. Smilansky, Gnutzmann, and
Weber \cite{GnuSmiWeb_wrm04} investigated the same question on a more
approachable model, a quantum graph, where they conjectured that
under some mild assumptions on the graph structure and edge lengths,
the number of zeros of the $k$-th eigenfunction should look like a
Gaussian shifted by $k-1$.  Alon, Band and Berkolaiko
\cite{AloBanBer_cmp17, AloBanBer_em22} provided a rigorous formulation
of the conjecture in the context of quantum graphs.

\quash{
For a quantum graph $(G,\ell)$,
where $G$ is the graph structure and $\ell$ is the vector of edge
lengths, the nodal surplus random variable is the limit
$\sigma(G,\ell):=\lim_{N\to\infty}\sigma_N(G,\ell)$, where
$\sigma_N(G,\ell)$ is the nodal surplus of the $k$-th eigenfunction of
$(G,\ell)$ for a random $k$ is chosen uniformly from $1,2,\ldots, N $.
It was shown in \cite{AloBanBer_cmp17, AloBanBer_em22} that
$\sigma(G,\ell)$ exists and is symmetric around its mean
$\mathbb{E}(\sigma(G,\ell))=\beta(G)/2$. It was also proven that if
the graph has disjoint cycles, then $\sigma(G,\ell)$ is the binomial
distribution Bin$(\beta(G),\frac{1}{2})$, which supports the
conjecture.  Denote the normalized nodal surplus by
 \begin{equation}\label{eqn-normalized-qg}
   \Sigma(G,\ell):=\frac{\sigma(G,\ell)-\beta(G)/2}{\mathrm{std}(\sigma(G,\ell))} \end{equation}
 In \cite{AloBanBer_em22}, the conjecture was formulated as follows: Given any sequence of connected graphs graphs $G_{j}$ with $\beta(G_{j})\to\infty$, consider all possible choices of $\QQ$-linearly independent $\ell_{j}$ for $G_{j}$, then
\[\sup_{\ell_{j}}\mathrm{dist} \left(\Sigma(G_{j},\ell_{j}) - N(0,1) \right) \longrightarrow 0,\]
where $\mathrm{dist}$ denotes the KS-distance, and  $N(0,1)$ denotes the normal distribution. Numerical evidence and analytical partial results were given in \cite{AloBanBer_em22}. 

\GB{I would reduce the amount of detail on the quantum graph formulation.}
}

The analogous question was asked in \cite{AloGor_jst23, AloGor_prep24}
for discrete graphs and matrices $h\in\cS(G)$, where it was shown to
be beneficial to consider $\sigma(h',k)$ for a random signing $h'$ of
$h$ as well as random $k$ from $1,\ldots,n$. 
It was conjectured in \cite{AloGor_jst23} that this random variable
$\sigma(h',k)$, for random signing and $k$, should have mean $\beta(G)/2$ and that the normalized nodal surplus distribution
\begin{equation}
  \label{eqn-normalized}
  \Sigma(h',k)
  :=
  \frac{\sigma(h',k)-\mathbb{E}(\sigma(h',k))}{\mathrm{std}(\sigma(h',k))}
\end{equation}
should converge to $N(0,1)$ as $\beta(G)\to\infty$.  This conjecture
was confirmed for graphs $G_{oo}$ with disjoint cycles\footnote{For quantum
  graphs with disjoint cycles this was established earlier, in
  \cite{AloBanBer_cmp17}.}: in \cite{AloGor_prep24} it was shown for
generic $h\in\cS(G_{oo})$ that the nodal surplus distribution is binomial:
\begin{equation}
  \label{eq:binomial}
  \sigma(h',k) \sim \textrm{Bin}\left(\beta(G_{oo}),\frac{1}{2}\right).   
\end{equation}

However, this conjecture is refuted by a new result of
\cite{AloMikUrs_prep24}: if $G$ is the complete graph and $h$ is a
random GOE matrix, then $\Sigma(h',k)$, for a random signing
\textbf{and a random $k$}, converges to the semi-circle (rather than
normal) distribution.

Nevertheless, there is overwhelming numerical evidence that if one
keeps $k$ fixed and takes the distribution over random signings only,
the conjecture is still valid.  The results of \cite{AloMikUrs_prep24}
are reconciled by observing that the mean and variance of such
distribution depend on $k$, see \Cref{fig:Mean_surplus_curves4}, and
merging the data over a set of $k$ creates a Gaussian mixture (which
can approximate any distribution).  To summarize, the following
refinement of our conjecture appears to hold in all known cases:

\begin{conjecture}
  \label{conj:nodal_universality}
  Given a graph G and a generic real
  symmetric matrix $h$ in $\cS(G)$, consider the random variable
  $\sigma(h',k)$ defined as the nodal surplus of the $k$-th
  eigenfunction of a random signing $h'$ of $h$. 

  Then, for any sequence of graphs $(G_j, h_j)$ with first Betti numbers
  $\beta(G_j) \to \infty$,
  \begin{equation}
    \label{eq:univ_conjecture}
    \sup_k \mathrm{dist}_{KS} \left( \Sigma (h_j',k), N(0,1) \right) \longrightarrow 0,    
  \end{equation}
  where $\mathrm{dist}_{KS}$ denotes the Kolmogorov--Smirnov distance,
  $\Sigma(h',k)$ is the normalized version of $\sigma(h',k)$, see
  \eqref{eqn-normalized}, and $N(0,1)$ denotes the standard normal
  distribution.
\end{conjecture}

\subsection{Experimenting with 3-regular graphs}
\label{sec-experiments-3reg}

If $G$ is $3$-regular and $V_{N}$ is an admissible support, then \Cref{subsec-admissible} says that any $r\in V_{ZN}$ has exactly three neighbors in $V_{N}$ and therefore no other neighbors. We conclude the following lemma:

\begin{lem}
  \label{lem:VN_3reg}
If $G$ is 3-regular, then a set of vertices $V_{N}$ is an admissible support if and only if the following two conditions hold:
\begin{enumerate}
    \item The graph $G_N$ induced by the vertex set $V_{N}$ is connected.
    \item The complement of $V_{N}$ is an independent set in $G$. That is, $E_{ZZ}=\emptyset$ and $V_{ZZ}=\emptyset$.
\end{enumerate}
\end{lem}

Next we characterize possible critical manifolds $F$, cf.\ \Cref{thm:main}. 

\begin{lem}\label{lem:F 3reg}
If $G$ is 3-regular, $h\in\cS(G)$ is generic (as in \ref{ass}), $V_{N}$ is an admissible support, and $\psi$ is an eigenvector of some signing of $h|G_N$, then for any $r\notin V_{N}$ define
\[
b_{r}=(|h_{rs}\psi_{s}|)_{s\sim r}\in \mathbb{R}_{+}^3.
\]
The corresponding critical manifold $F$ then has two possibilities:
\begin{enumerate}
    \item If some 3-vector $b_{r}$ does not satisfy the triangle inequality \eqref{eqn-triangle}, then $F=\emptyset$;
    \item If every 3-vector $b_{r}$ satisfies the triangle inequality, then $F$ consists of exactly $2^{n-|V_{N}|}$ points.
\end{enumerate}
\end{lem}

It is only necessary to check $2^{\beta(G_N)}$ signings for each
choice of $G_N$, rather than $2^{|E_{NN}|}$ signings.  As a
consequence of \Cref{lem:VN_3reg}, $\beta(G_N) = \beta(G) - 2
(n-|V_N|)$.

\begin{figure}
    \centering
    \includegraphics[width=0.6\textwidth]{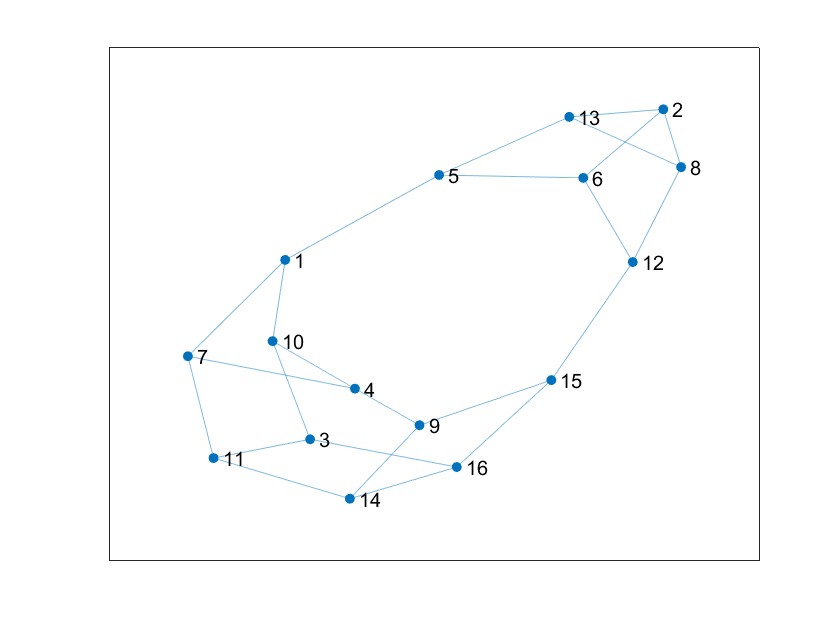}
    \caption{An example of one of the randomly chosen 3-regular graphs with $n=16$.}
    \label{fig:graph_n16_1}
\end{figure}

Two experiments were conducted. In the first, we considered three random
3-regular graphs with $n=16$, $18$ and $20$, correspondingly; The
$n=16$ graph is shown in Figure~\ref{fig:graph_n16_1} for one
example. We considered matrices of the form $h=A+D$, where $A$ is the
adjacency matrix and $D$ is diagonal with
\begin{equation}
  \label{eq:diag_entries}
  D_{ss}=\sqrt{2}s \mod{1},
\end{equation}
chosen to minimize the chances of $h$ being non-generic.
The number of admissible supports
found was found to be $\num{566}$, $\num{1292}$, and $\num{2602}$,
respectively.  For each admissible $V_{N}$ and each signing of
$h|G_N$, the eigenvectors were computed and checked for compliance
with the triangle inequality, see \Cref{lem:F 3reg}.

\begin{figure}
  \centering
  \includegraphics[scale=0.75]{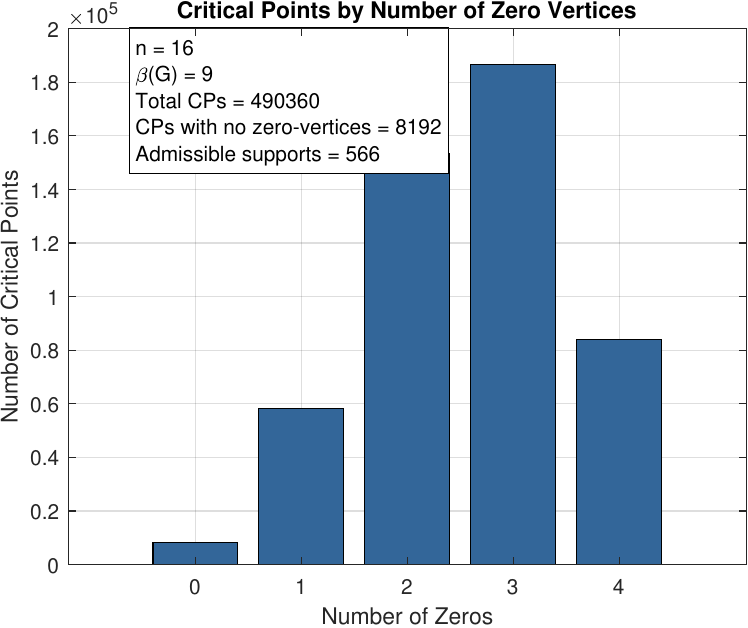}
  \caption{Number of critical points in a 3-regular graph against the
    size of the zero set $V \setminus V_N$ of the corresponding
    eigenvector (the symmetry critical points are the 0-th bar).}
  \label{fig:CP_hist}
\end{figure}

As a result, $\num{490360}$, $\num{1582176}$, and $\num{5908840}$
critical points were identfied in the three graphs, as compared to
$n2^{\beta(G)}=\num{8192}$, $\num{18432}$ and $\num{40960}$ symmetric
critical points, respectively.  For the graph with $n=16$, we further
classify the critical points by the size of their support (or
equivalently, the zero set $V \setminus V_N$), presenting the
histogram in \Cref{fig:CP_hist}.  We conclude that critical points due
to eigenvectors with vanishing entries appear in
great abundance.

In the second experiment, we aimed to examine larger graphs with higher values of $n$ and to range over more graphs. Here, we calculated the nodal count for all eigenvectors of all signings. Due to Theorem 3.2 in \cite{AloGor_jst23}, we know that if $\lambda_{k}$ has no other critical points except for the symmetry points, then the nodal surplus of a random signing $\sigma(h',k)$ is binomial with mean $\mathbb{E}(\sigma(h',k))=\beta(G)/2$. As a result, whenever $\mathbb{E}(\sigma(h',k))\ne\beta(G)/2$, we know that there must be other critical points of $\lambda_{k}$.

In this experiment, we considered $n=10,14,18,22,26,30,34,38,42$. For each $n$, we chose five random graphs and calculated the nodal count for all eigenfunctions of all signings. In Figure~\ref{fig:Mean_surplus_curves4}, we plot ${\mathbb{E}(\sigma(h',k))}/{\beta(G)}$ against $k/n$ (both bounded between $0$ and $1$) for all 35 graph samples with $n\ge 18$. Clearly, these points concentrate around a curve far from the constant line ${\mathbb{E}(\sigma(h',k))}/{\beta(G)}\equiv \frac{1}{2}$.

 Figure~\ref{fig:MaxKS_vs_n} provides evidence supporting  the conjecture that
\[
\sup_k \mathrm{dist}_{KS}\left( \Sigma(h_j',k), N(0,1) \right)\to 0,\quad\text{as}\quad \beta(G)\to\infty,
\]
where the distance is the Kolmogorov–-Smirnov distance, and
$\Sigma(h',k)$ is defined in \eqref{eqn-normalized}.  We plot the
values of $\sup_k \mathrm{dist}\left( \Sigma(h_j',k), N(0,1) \right)$
for each graph as a function of $n$. The plot strongly suggests these
distances tend uniformly to zero as $n$ grows.

\begin{figure}
    \centering
    \includegraphics[width=0.75\textwidth]{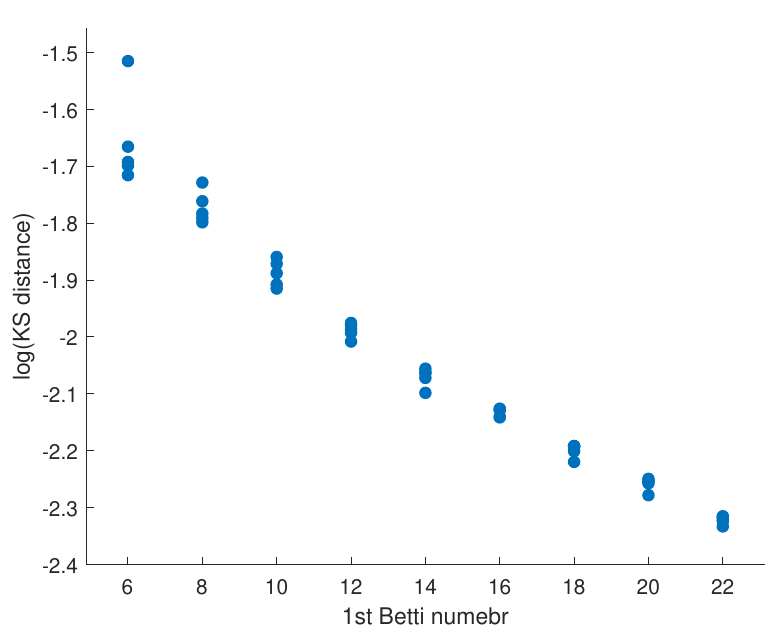}
    \caption{Plot of the logarithm of the maximal distance $\max_k
      \mathrm{dist} \left( \Sigma (h',k), N(0,1) \right)$ against $\beta(G)=\frac{n}{2}+1$, for $45$  randomly chosen 3-regular
      graphs, $5$ graphs for each $n\in\{10,14,18,\ldots,42\}$, supporting the conjecture. }
    \label{fig:MaxKS_vs_n}
  \end{figure}

\appendix

\section{Planar linkages}
\label{sec:linkage}

\subsection{Planar linkages}
\label{subsec-linkage}
Let $b = (b_1, b_2, \cdots, b_d) \in \RR^d$ with each $b_j > 0$ (abbreviated $b>0$).  
Consider the set of solutions $(\theta_1, \theta_2,
\cdots, \theta_d) \in (S^1)^d$ of the {\em planar linkage equation}
\begin{equation}\label{eqn-linkage} \sum_{j=1}^d b_j e^{i \theta_j} = 0\end{equation}
which arises in the context of this paper as the eigenvalue equation
$(h_\alpha \psi)_r = 0$ at a vertex $r$ with $\psi_r = 0$.  Equation
\eqref{eqn-linkage} represents a collection of $d$ vectors in
$\RR^2 = \CC$, of lengths $b_j$, placed tail to head, beginning and
ending at the origin, that is, a {\em planar linkage}.  The {\em
  pre-linkage set} $\widetilde{M}(b)$ is the set of solutions
$\theta \in (S^1)^d$.  The set of solutions, modulo rotations of the
whole linkage, is a semialgebraic set denoted $M(b)$
\cite{FarFro_tams13,FarSch_gd07,Walker_thesis85} and referred to as
the planar linkage space attached to $b$.  Equations for $M(b)$ in the
quotient $(S^1)^d/S^1 \cong (S^1)^{d-1}$ (after dividing by rotations)
may be obtained by choosing the rotation equal to $-\theta_d$ so that
the last linkage lies on the real axis.  Then $M(b)$ is the zero set
of the function $f: (S^1)^{d-1} \to \RR$,
\begin{equation}\label{eqn-Mb}
f(\theta) := \sum_{j=1}^{d-1} b_j e^{i \theta_j} + -b_d.
\end{equation}

Let us say that $b \in \RR^d_{>0}$ is \emph{generic} if 
\begin{equation}\label{eqn-generic-linkage}
\sum_{j=1}^d \epsilon_j b_j \ne 0  \text{ for every } \epsilon_j \in \{ +1, -1\}. \end{equation}
Let $(M_1, M_2, \ldots, M_d)$ be the a permutation of $(b_1,\ldots,
b_d)$ ordered in decreasing order (in particular, $M_s = \max b_j$).  Let $\beta_0(b)$ denote the number of connected components of $M(b)$.

\begin{thm}
  \label{thm-linkage}
Suppose $b>0$ satisfies the generic linkage equation \eqref{eqn-generic-linkage}. 
Then $\beta_0(b) \in \{0,1,2\}$.
The linkage manifold $M(b)$ is non empty ($\beta_0(b) \ne 0$) if and only if $d \ge 3$ and
\begin{equation}\label{eqn-triangle}
M_1 < \frac12 \sum_r b_r. \end{equation}  
In this case,
\begin{enumerate}
\item \label{item:surjective_differential}
  The complex valued map $f$ defining $M(b)$, equation~\eqref{eqn-Mb},
  has a surjective differential $df_\theta:\RR^{d-1}\to\CC$ at every point $\theta\in M(b)$. 
\item \label{item-nonempty} the linkage space $M(b)$ is an analytic manifold of dimension $d-3$.
\item \label{item:connected_components}
  If $M_2 + M_3 < \frac{1}{2}\sum_r b_r$ then $M(b)$ is connected ($\beta_0(b) = 1$).  Otherwise it has two connected components ($\beta_0(b) = 2$),
  exchanged by complex conjugation, each diffeomorphic to the torus $\TT^{d-3}$.
\item \label{item:continuous} If $b' >0$ is sufficiently close to $b$ then the manifolds $M(b)$, $M(b')$ are diffeomorphic.
\end{enumerate}
\end{thm}
\begin{proof}
If the largest link is greater than the sum of the others then the linkage cannot close so there
is no solution.  (This also implies $d \ge 3$ since the only solution with $d = 2$ fails the
genericity assumption). The defining complex function $f$ can be replaced by the real vector valued function $f_{\RR}=(\Re(f),\Im(f))$ which is real analytic, so $M(b)$ is a real analytic variety. To conclude both (\ref{item:surjective_differential}) and $\dim(M(b))\le (d-1)-2=d-3$, we need to show that the differentials of the real and imaginary parts $d(\Re(f(\theta)))=\Re( df(\theta)) $ and $d(\Im(f(\theta)))=\Im(df(\theta)) $ are linearly independent over $\R$ at every point $\theta\in M(b)$. We calculate
\begin{align}
\label{eqn-dF}
d f_\theta &= i\left(b_1 e^{i\theta_1}, b_2 e^{i \theta_2}, \cdots, b_{d-1}e^{i \theta_{d-1}} \right)\\
\nonumber
&= i e^{i\theta_1}\left(b_1, b_2e^{i(\theta_2-\theta_1)}, \cdots, b_{d-1}e^{i(\theta_{d-1}-\theta_1)}\right).\end{align}
Since $b_1\neq0$, the real and imaginary parts are linearly dependent
if and only if the vector $\left(b_1, b_2e^{i(\theta_2-\theta_1)},
  \cdots, b_{d-1}e^{i(\theta_{d-1}-\theta_1)}\right)$ is real, which
means $e^{i(\theta_j-\theta_1)}\in\{-1,1\}$ for all $j$.  Defining
$\epsilon_{j}=e^{i(\theta_j-\theta_1)}$, the equation $f(\theta) = 0$ becomes
\[ f(\theta)= e^{i\theta_1} \sum_{j=1}^{d-1} \epsilon_{j} b_j +
  b_d=0,\]
implying $e^{i\theta_1}\in\{-1,1\}$ and therefore violating the genericity condition.

Item (3) in the theorem is proven in \cite{KapMil_jdg95,FarSch_gd07,HauKnu_aif98}.
For part (4) it suffices to consider the case when $b, b'$ differ in a single coordinate which we may
take to be the last coordinate.  If $b_d$ is sufficiently close to $b'_d$ then
equation \eqref{eqn-Mb} and the implicit function theorem
gives a diffeomorphism $M(b) \cong M(b')$.
\end{proof}

\section{The derivatives}\label{sec-derivatives}
\subsection{}\label{subsec-criticality}

We summarize several calculations from \cite{AloGor_jst23}, see also \cite{BerCanCoxMar_paa22}. Let $h \in \cH(G)$.  Assume $\lambda_k$
is a simple eigenvalue of $h$.  Consider varying $h$
within the torus of magnetic perturbations, in the analytic 1-parameter family $h(t)=(t\alpha *h)$.  
There is an associated analytic family of normalized eigenvectors $\psi(t)$
 with $h_t\psi(t) = \lambda_k(t)\psi(t)$. Using Leibniz dot notation for
the derivative at $t=0$, we have $\dot h \psi + h \dot\psi = \dot\lambda \psi + \lambda \dot \psi$.
So $(\dot h)_{rs} = i \alpha_{rs}h_{rs}$ and
\begin{equation}\label{eqn-derivative}
    \dot \lambda= \left< \dot h \psi, \psi\right> = 2 \sum_{r \sim s}  \alpha_{rs}\Im ( h_{rs} \bar\psi_r \psi_s).
    \end{equation}
Therefore $h$ is a critical point of $\lambda_k$ if and only if the following {\em criticality
condition} holds,
\[ h_{rs} \bar\psi_r \psi_s \in \RR \ \text{ for all } r \sim s.\]
\subsection{}
To focus on the variation of $\lambda_k$ with respect to $\alpha$,
let $\mu(\alpha) = \lambda_k(h_\alpha)$ and assume $\alpha = 0$ is a critical point of $\mu$.
To determine the second derivative consider two tangent vectors
 $\gamma,\delta \in T_{\alpha} \cA(G) =\cA(G)$ and the corresponding two parameter family
$h(s,t) = (s\gamma + t\delta)*h$.  Set
\[ \mu(s,t) = \lambda_k((s\gamma + t \delta)*h).\]
Abusing notation, we use $\partial_{\gamma}$ and $\partial_{\delta}$ to denote derivatives in $s, t$ at $s = t = 0$, that is,
\[ \partial_{\gamma}h = \frac{d}{ds}(s\gamma *h)\vert_{s=0}\ \text{ with }\ (\partial_{\gamma}h)_{rs} = i \gamma_{rs}h_{rs}.\]  

The Hessian of $\mu$ was computed in \cite{AloGor_jst23} (5.3), 
\[ \left< \gamma, \Hess(\mu)\delta \right> = 2 \Re(\left<\partial_{\gamma}\psi, (\partial_{\delta} h) \psi \right>)
+ \left< \psi, (\partial^2_{\gamma,\delta}h)\psi\right>\]
Using $\left< \psi, \psi \right> =1$ and
$(h-\lambda)\psi = 0$ gives 
\begin{equation} 0=\partial_{\delta}\lambda = \left< \psi, (\partial_{\delta}h) \psi \right> \end{equation} and 
\[ (\partial_{\gamma}h) \psi + (h-\lambda) \partial_{\gamma}\psi = 0. \]
Therefore $\partial_{\gamma}\psi + (h-\lambda)^+(\partial_{\gamma}h )\psi = c(\gamma) \psi$
for some constant $c(\gamma)$ since $\ker(h-\lambda)$ is 1-dimensional and spanned by $\psi$.  Taking the inner product with $(\partial_{\delta}h) \psi$ leaves
\[ \left<(\partial_{\delta}h) \psi, \partial_{\gamma} \psi\right> = - \left<(\partial_{\delta}h) \psi, (h-\lambda)^+ (\partial_{\gamma} h) \psi \right>.\]
 We obtain for the Hessian,
\begin{equation}\label{eqn-Hessian1}
\left< \gamma, \Hess(\mu) \delta \right> = -2\Re \left< B\delta, (h-\lambda)^+B\gamma \right> + \left< \psi, \partial^2_{\gamma,\delta}h \psi \right>\end{equation}
where $B$ is defined in equation \eqref{eqn-B}.

\section{The real part of a Hermitian form}

\begin{prop}\label{prop-Qreal}
  Let $H \in \cH_n(\CC)$ be Hermitian and let $B:\RR^m \to \CC^n$ be a
  real linear surjective mapping.  For $x, y \in \RR^m$ let
  \[ Q(x,y) = \Re\left< Bx, HBy\right>_{\CC},\]
where $\left< *, * \right>_{\CC}$ denotes the standard Hermitian form on $\CC^n$.  Then the inertia index of this form is:
\begin{align}
  \label{eq:Qreal_index}
  n_-(Q) & = 2 n_-(H) \\
  \label{eq:Qreal_positive}
  n_+(Q) & = 2n_+(H)  \\
  \label{eq:Qreal_nullity}
  n_0(Q) &= 2n_0(H) + m - 2n.
\end{align}
\end{prop}

\begin{proof}
  Expand $H = H_0 + i H_1$ into its (symmetric) real and
  (skew-symmetric) imaginary parts.  By choosing appropriate
  coordinates\footnote{By Sylvester's law of inertia, indices of $Q$
    do not depend on the choice of basis.} on $\RR^m$ we can assume
  \begin{equation}
    \label{eq:B_choice_basis}
    B(u\oplus v\oplus w) = u+ iv
    \qquad\text{for}\quad
    u\oplus v\oplus w \in \RR^n \oplus \RR^n \oplus \RR^{m-2n}.
  \end{equation}
  In these coordinates, $Q$
  is represented by the symmetric matrix
  \[
    Q =
    \begin{pmatrix}
      H_0 & -H_1 & 0\\ H_1 & H_0 & 0\\ 0&0&0
    \end{pmatrix}
  \]
  If $\psi=\xi + i\eta \in \CC^n$ is an eigenvector of $H$ with
  eigenvalue $\lambda \in \RR$ then
  \[
    R\psi = \begin{pmatrix}\xi \\ \eta \\ 0 \end{pmatrix}\in \RR^m
    \ \text{ and }\
    R(i\psi) = \begin{pmatrix} -\eta \\ \xi \\ 0 \end{pmatrix}\in \RR^m
  \]
  are (orthogonal) eigenvectors of $Q$ with the same eigenvalue.
  Furthermore, if
  $\left\{ \psi_1, \ldots, \psi_n \right\} \subset \CC^n$ are
  linearly independent (resp.~orthogonal) over $\CC$ then
  $\left\{R\psi_1, R(i\psi_1), \ldots, R\psi_n,
    R(i\psi_n)\right\}$ are linearly independent (resp. orthogonal)
  over $\RR$.  Together with a basis of $0\oplus0\oplus \RR^{m-2n}$
  they span $\RR^m$ and equations
  \eqref{eq:Qreal_index}-\eqref{eq:Qreal_nullity} follow.
\end{proof}

We remark that the imaginary part $\Im \left< Bx, HBy\right>_{\CC}$ is
skew-symmetric with purely imaginary eigenvalues which may be distinct
but come in complex conjugate pairs.

\bibliography{bk_bibl}
\bibliographystyle{siam}
\end{document}